\documentclass{fancy-article}
\pdfoutput=1

\usepackage{tikz}
\usetikzlibrary{cd}
\usetikzlibrary{calc, positioning, decorations.pathreplacing}

\usepackage[xcolor,hyperref,cleveref,notion,quotation]{knowledge}
\knowledgeconfigure{quotation, protect quotation={tikzcd}}
\knowledgeconfigure{diagnose line=true, diagnose bar=true}
\knowledgeconfigure{anchor point color={Ruby Red}, anchor point shape=corner}
\knowledgeconfigure{visible anchor points=false}
\knowledgestyle{intro notion}{color={Ruby Red}, emphasize}
\knowledgestyle{notion}{color={Blue Sapphire}}
\hypersetup{
    colorlinks,
    breaklinks=true,
    linkcolor={Blue Sapphire},
    citecolor={Blue Sapphire},
    urlcolor={Blue Sapphire}
}
\knowledge{ie}[i.e.]{italic, text={{i.e.}}}
\knowledge{Ie}[I.e.]{italic, text={{I.e.}}}

\knowledge{notion}
 | preorder
 | preorders
 | Preorders

\knowledge{notion}
 | algorithm

\knowledge{notion}
 | yes

\knowledge{notion}
 | no

\knowledge{notion}
 | separation problem
 | separation problems

\knowledge{notion}
 | witness function
 
\knowledge{notion}
 | \FO-separator
 | \FO-separator sentence
 | first-order separator
 | first-order separator sentence
 | \FO-separability
 | \FO-inseparability
 | \FO-separation

\knowledge{notion}
 | separate
 | separates
 | separating
 | \FO-separate
 | \FO-separates
 | \FO-separating
 | \FO-separated
 | \FO-separable
 | separable

 \knowledge{notion}
 | saturation
 | saturations
 | Saturations

\knowledge{notion}
  | variable
  | variables
  | Variables
  
\knowledge{notion}
 | free variable
 | free variables
 | Free variables

\knowledge{notion}
  | valuation
  | valuations
  | Valuations

\knowledge{notion}
  |	monadic second-order logic
  |	monadic logic
  |	MSO

  \knowledge{notion}
  |	weak monadic second-order logic
  |	weak monadic logic
  |	WMSO

 \knowledge{notion}
  | first-order logic
  | \FO -logic
\knowledge{notion}
  | satisfies
\knowledge{notion}
  | accepts

 \knowledge{notion}
  | \FO -condensation
  | \FO -condensations
  | condensation formula
  | condensation \FO -formula

\knowledge{notion}
  | defined@lang
  | \FO -definable@lang
  | \FO -definability@lang
  | \FO -definable language
  | \FO -definable languages
  | first-order definable languages
  | definable in first-order logic

\knowledge{notion}
  | \FO -definable

 \knowledge{notion}
  | \FO -definable@map
  | \FO -definability@map
  | \FO -definable map
  | \FO -definable maps
  | \FO -definable function
  | \FO -definable functions
 | first-order definable maps

 \knowledge{notion}
  | Bedon's theorem

\knowledge{notion}
  | \FO -approximant
  | \FO -approximants
  | {\FO }-approximants

 \knowledge{notion}
  | finite condensation

\knowledge{notion}
 | first-order formula
 | \FO -formula
 | \FO -formulae

\knowledge{notion}
 | \FO -sentence
 | \FO -sentences

\knowledge{notion}
 | \FOk-equivalent

\knowledge{notion}
 | quantifier depth

\knowledge{notion}
 | \FO-indistinguishable
 | first-order indistinguishability

\knowledge{notion}
 | \FOk -closure

\knowledge{notion}
 | semigroup
 | semigroups
 | Semigroups

\knowledge{notion}
 | monoid
 | monoids
 | Monoids

\knowledge{notion}
  | idempotent

\knowledge{notion}
 | ordinal monoid
 | ordinal monoids
 | Ordinal monoids
 | finite ordinal monoid
 | finite ordinal monoids
 | Finite ordinal monoids

\knowledge{notion}
 | presentation of a finite ordinal monoid
 | presentation@OM

\knowledge{notion}
 | recognising@OM
 | recognise@OM
 | recognised@OM
 | recognises@OM
 | recognition@OM
 | recognisable@OM

\knowledge{notion}
 | ordered ordinal monoid
 | ordered ordinal monoids
 | Ordered ordinal monoids

\knowledge{notion}
 | aperiodic merge
 | merge operator
 | merge operators
 | Merge operators

\knowledge{notion}
 | ordinal monoid with merge
 | ordinal monoids with merge
 | Ordinal monoids with merge

\knowledge{notion}
 | generalised product

\knowledge{notion}
 | generalised associativity
 | Generalised associativity

\knowledge{notion}
 | idempotent power
 | idempotent powers

\knowledge{notion}
 | ordinal monoid morphism
 | ordinal monoid morphisms
 | Ordinal monoid morphisms
 | morphism of ordinal monoid
 | morphism of ordinal monoids
 | morphism@OM
 | morphisms@OM
 | Morphisms@OM

\knowledge{notion}
  | ordinal submonoid

\knowledge{notion}
  | ordinal subsemigroup

\knowledge{notion}
  | ordinal monoid congruence
  | ordinal semigroup congruence
  | congruence@OM

\knowledge{notion}
  | power ordinal semigroup

\knowledge{notion}
  | power ordinal monoid

\knowledge{notion}
 | \FO-pointlike sets
 | \FO-pointlikes
 | aperiodic pointlike
 | aperiodic pointlikes
 | pointlike set
 | pointlike sets
 | pointlikes
 | pointlike
 | aperiodic pointlike sets
 | Pointlike sets
 
\knowledge{notion}
 | aperiodic
 | aperiodicity

\knowledge{}
 | syntactic monoid

\knowledge{notion}
 | linear ordering
 | linear orderings
 | Linear orderings
 
 \knowledge{notion}
  | finite@linord

\knowledge{notion}
 | countable linear ordering
 | countable linear orderings
 | Countable linear orderings
 | countable@linord
 | countable ordinal
 | countable ordinals
 | Countable ordinals

\knowledge{notion}
 | scattered@linord

\knowledge{notion}
 | morphism@linord

\knowledge{notion}
 | isomorphism@linord
 
\knowledge{notion}
 | embedding@linord

\knowledge{notion}
 | sum@linord

\knowledge{notion}
 | product@linord

\knowledge{notion}
 | well-founded

\knowledge{notion}
 | ordinal
 | ordinals
 | Ordinals
 | ordinal@linord

\knowledge{notion}
 | alphabet
 | alphabets
 | Alphabets
 
\knowledge{notion}
 | word
 | words
 | Words
 
\knowledge{notion}
 | regular languages of countable ordinal words
 | regular languages@COW
 | regular languages@cow
 | regular language@cow
 | regular@cow
  
\knowledge{notion}
 | finite word
 | finite words
 | Finite words
 | finite@word
 
\knowledge{notion}
 | $\omega$-word
 | $\omega$-words
 | \(\omega\)-word
 | \(\omega\)-words
 | word of length~$\omega $
 | words of length~$\omega $
 | Words of length~$\omega $
 
\knowledge{notion}
 | domain

\knowledge{notion}
 | scattered word
 | scattered words
 | Scattered words
 | scattered@word

\knowledge{notion}
 | word@ord
 | words@ord
 | countable ordinal words
 | Countable ordinal words 
 | countable ordinal word
 | word of countable ordinal length
 | words of countable ordinal length
 | Words of countable ordinal length
 | ordinal word
 | ordinal words
 | Ordinal words
 | ordinal@word
 
\knowledge{notion}
 | countable word
 | countable words
 | Countable words
 | countable@word

\knowledge{notion}
 | concatenation

\knowledge{notion}
 | successor ordinal
 | successor ordinals
 | Successor ordinals
 
\knowledge{notion}
 | limit ordinal
 | limit ordinals
 | Limit ordinals

\knowledge{notion}
 | omega iteration
 | \(\omega\)-iteration
 | \(\omega\)-iterations
 | $\omega$-iteration
 | $\omega$-iterations

\knowledge{notion}
 | condensation
 | condensations
 | Condensations

\knowledge{notion}
 | equivalence relation associated with
 | corresponding equivalence relations

\knowledge{notion}
 | Green's relations
 | Green's relation

\knowledge{link=\Jleq}
 | $\greenJ $-smaller or equal
 
\knowledge{notion}
 | $\Jeq $-class
 | $\Jeq $-classes

\knowledge{notion}
 | $\Req $-class
 | $\Req $-classes

\knowledge{notion}
 | $\Leq $-class
 | $\Leq $-classes

\knowledge{notion}
 | $\Heq $-class
 | $\Heq $-classes

\knowledge{notion}
 | $\Deq $-class
 | $\Deq $-classes

\knowledge{notion}
 | regular@Jclass

\knowledge{notion}
 | $\omega $-stable
 | \(\omega \)-stable
 
\knowledge{notion}
 | $\Heq $-trivial
 | \(\Heq \)-trivial
 | $\Leq $-trivial
 | $\Leq $-triviality
 | \(\Jeq\)-trivial
 | $\Jeq$-trivial
 | greentrivial

\knowledge{notion}
 | group-trivial

\knowledge{notion}
 | covering problem for regular languages of ordinal words
 | covering problem
 | covering problems
 | \FO-covering problem
 | \FO-covering

\knowledge{notion}
 | algorithm@cov
 
\knowledge{notion}
 | yes@cov
 
\knowledge{notion}
 | no@cov

\newcommand\labelwithproof[1]{\label{#1}\marginnote{\footnotesize{See the proof of \Cref{#1} at page~\pageref{proof-#1}.}}}

\newrobustcmd\pto{\rightharpoonup}

\newrobustcmd\Nats{\mathbb{N}}
\newrobustcmd\Zed{\mathbb{Z}}
\newrobustcmd\Ints{\mathbb{Z}}
\newrobustcmd\powerset{\mathcal P}

\knowledgenewrobustcmd\down{\cmdkl{\downarrow}}

\newrobustcmd\flatten{\mathop{\kl[\flatten]{\textnormal{flat}}}}
\knowledge\flatten{notion}

\newrobustcmd\evalexp{{\kl[\evalexp]{\textnormal{eval}}_{\lettermap}}}
\knowledge\evalexp{notion}
\newrobustcmd\unfoldexp{{\kl[\unfoldexp]{\textnormal{unfold}}}}
\knowledge\unfoldexp{notion}

\newrobustcmd\alphabet{\kl[\alphabet]{\Sigma}}
\knowledge\alphabet{notion}

\newrobustcmd\dom{\mathop{\kl[\dom]{\textnormal{dom}}}}
\knowledge\dom{notion}
\newrobustcmd\finite{\mathop{\kl[\finite]{\textnormal{finite}}}}
\knowledge\finite{notion}
\newrobustcmd\isSuccessor{\mathop{\kl[\isSuccessor]{\textnormal{isSuccessor}}}}
\knowledge\isSuccessor{notion}
\newrobustcmd\isLimit{\mathop{\kl[\isLimit]{\textnormal{isLimit}}}}
\knowledge\isLimit{notion}
\newrobustcmd\last{\mathop{\kl[\last]{\textnormal{last}}}}
\knowledge\last{notion}

\newrobustcmd\ord{^{\kl[\ord]{\textnormal{ord}}}}
\knowledge\ord{notion}
\newrobustcmd\ordp{^{\kl[\ordplus]{\textnormal{ord{+}}}}}
\knowledge\ordplus{notion}
\newrobustcmd\Womega{^{\kl[\Womega]{\omega}}}
\knowledge\Womega{notion}

\newrobustcmd\FOkeq{\mathrel{\kl[\FOkeq]{\equiv_{\FOk}}}}
\knowledge\FOkeq{notion}

\newrobustcmd\unit{\kl[\unit]{1}}
\knowledge\unit{notion}

\newrobustcmd\condensation{\mathrel{\kl[\condensation]{\sim}}}
\knowledge\condensation{notion}

\newrobustcmd\fincondensation[1]{\mathrel{\kl[\fincondensation]{\sim_{\textnormal{fin},#1}}}}
\knowledge\fincondensation{notion}

\newrobustcmd\pluslinord{\mathrel{\kl[\pluslinord]+}}
\knowledge\pluslinord{notion}
\newrobustcmd\timeslinord{\mathrel{\kl[\timeslinord]\cdot}}
\knowledge\timeslinord{notion}
\newrobustcmd\plusord{\pluslinord}
\knowledge\plusord{notion}
\newrobustcmd\timesord{\timeslinord}
\knowledge\timesord{notion}

\newrobustcmd\leqord{\mathrel{\kl[\leqord]\preccurlyeq}}
\newrobustcmd\geqord{\mathrel{\kl[\leqord]\succcurlyeq}}
\knowledge\leqord{notion}
\newrobustcmd\lessord{\mathrel{\kl[\lessord]\prec}}
\newrobustcmd\greaterord{\mathrel{\kl[\lessord]\succ}}
\knowledge\lessord{notion}

\newrobustcmd\uncountord{\kl[\uncountord]\omega_1}
\knowledge\uncountord{notion}
\newrobustcmd\expord[2]{\kl[\expord]{{#1}^{#2}}}
\knowledge\expord{notion}
\newrobustcmd\lub{\kl[\lub]\bigvee}
\knowledge\lub{notion}

\newrobustcmd\factp[1][]{^{\kl[\factp]{n!#1}}}
\knowledge\factp{notion}

\newrobustcmd\monoid{\kl[\monoid]{\mathcal{M}}}
\knowledge\monoid{}%{notion}

\newrobustcmd\monoidmerge{\kl[\monoidmerge]{\mathcal{M}}}
\knowledge\monoidmerge{notion}

\newrobustcmd\ordsemigroup{\kl[\ordsemigroup]{\mathcal{S}}}
\knowledge\ordsemigroup{notion}
\newrobustcmd\ordmonoid{\kl[\ordmonoid]{\mathcal{M}}}
\knowledge\ordmonoid{notion}

\newrobustcmd\Aomega{^{\kl[\Aomega]{\underline{\omega}}}}
\knowledge\Aomega{notion}

\newrobustcmd\Acdot{\mathbin{\kl[\Acdot]{\underline{\cdot}}}}
\knowledge\Acdot{notion}
\newrobustcmd\Aunit{\kl[\Aunit]{\underline{1}}}
\knowledge\Aunit{notion}

\newrobustcmd\product{\kl[\product]{\pi}}
\knowledge\product{notion}

\newrobustcmd\idem[1][]{^{\kl[\idem]{\textnormal{idem}#1}}}
\knowledge\idem{notion}

\newrobustcmd\Pset{\mathop{\kl[\Pset]{\mathcal{P}}}}
\knowledge\Pset{notion}
\newrobustcmd\Pcdot{\mathrel{\kl[\Pcdot]{\underline\cdot}}}
\knowledge\Pcdot{notion}
\newrobustcmd\Pproduct{\kl[\Pproduct]{\pi}}
\knowledge\Pproduct{notion}
\newrobustcmd\Pomega{^{\kl[\Pomega]{\underline\omega}}}
\knowledge\Pomega{notion}
\newrobustcmd\Punit{\kl[\Punit]{\underline1}}
\knowledge\Punit{notion}

\newrobustcmd\PFOproduct{\kl[\PFOproduct]{\rho}} % FO-approximation of Pproduct
\knowledge\PFOproduct{notion}

\knowledgenewcommand\Closgo[1]{\cmdkl{\langle}#1\cmdkl{\rangle}^{\cmdkl{\textnormal{grp}{,}\omega}}}
\knowledgenewcommand\Closgord[1]{\cmdkl\langle#1\cmdkl\rangle^{\cmdkl{\textnormal{grp{,}ord}}}}
\knowledgenewcommand\Closgordp[1]{\cmdkl\langle#1\cmdkl\rangle^{\cmdkl{\textnormal{grp{,}ord+}}}}

\knowledgenewcommand\Closgp[1]{\cmdkl\langle#1\cmdkl\rangle^{\cmdkl{\textnormal{grp{+}}}}}
\knowledgenewcommand\Closgs[1]{\cmdkl\langle#1\cmdkl\rangle^{\cmdkl{\textnormal{grp{*}}}}}

\knowledgenewcommand\Closp[1]{\cmdkl\langle#1\cmdkl\rangle^{\cmdkl{+}}}
\knowledgenewcommand\Closs[1]{\cmdkl\langle#1\cmdkl\rangle^{\cmdkl{*}}}
\knowledgenewcommand\Closo[1]{\cmdkl\langle#1\cmdkl\rangle^{\cmdkl{\omega}}}
\knowledgenewcommand\Closord[1]{\cmdkl\langle#1\cmdkl\rangle^{\cmdkl{\textnormal{ord}}}}
\knowledgenewcommand\Closordp[1]{\cmdkl\langle#1\cmdkl\rangle^{\cmdkl{\textnormal{ord{+}}}}}

\newrobustcmd\Var{\kl[\Var]{\textnormal{Var}}}
\knowledge\Var{notion}

\newrobustcmd\valuation{\kl[\valuation]{\nu}}
\knowledge\valuation{notion}

\newrobustcmd\FO{\ensuremath{\kl[\FO-formula]{\textnormal{FO}}}}
\newrobustcmd\FOk{\ensuremath{\textnormal{FO}_k}}

\newrobustcmd\closureFOk[1]{\kl[\closureFOk]{[#1]_{\FOk}}}
\knowledge\closureFOk{notion}

\knowledgenewcommand\Sat{\cmdkl{\textnormal{Sat}}}

\newrobustcmd\cupfun{\mathrel{\kl[\cupfun]\cup}}
\knowledge\cupfun{notion}
\newrobustcmd\concatfun{\mathrel{\kl[\concatfun]{\odot}}}
\knowledge\concatfun{notion}
\newrobustcmd\focomp[1]{\mathrel{\kl[\focomp]{\circ_{#1}}}}
\knowledge\focomp{notion}

\newrobustcmd\PL{\kl[\PL]{\textnormal{Pl}}_{\FO}}
\knowledge\PL{notion}
\newrobustcmd\PLp{\kl[\PLp]{\textnormal{Pl^{+}}}}
\knowledge\PLp{notion}
\newrobustcmd\PLs{\kl[\PLs]{\textnormal{Pl^{*}}}}
\knowledge\PLs{notion}
\newrobustcmd\PLo{\kl[\PLo]{\textnormal{Pl^{\omega}}}}
\knowledge\PLo{notion}
\newrobustcmd\PLord{\kl[\PLord]{\textnormal{Pl^{\textnormal{ord}}}}}
\knowledge\PLord{notion}
\newrobustcmd\PLordp{\kl[\PLord]{\textnormal{Pl^{\textnormal{ord{+}}}}}}
\knowledge\PLordp{notion}

\newrobustcmd\Pmerge{^{\kl[\Pmerge]{\textnormal{grp}}}}
\knowledge\Pmerge{notion}

\newrobustcmd\Amerge{^{\kl[\Amerge]{\textnormal{grp}}}}
\knowledge\Amerge{notion}

\newrobustcmd\greenJ{\mathcal{J}}
\newrobustcmd\Jeq{\mathrel{\kl[\Jeq]{\greenJ}}}
\knowledge\Jeq{notion}
\newrobustcmd\Jleq{\mathrel{\kl[\Jleq]\leqslant_{\kl[\Jleq]{\greenJ}}}}
\knowledge\Jleq{notion}
\newrobustcmd\Jgeq{\mathrel{\kl[\Jgeq]\geqslant_{\kl[\Jgeq]{\greenJ}}}}
\knowledge\Jgeq{notion}
\newrobustcmd\Jl{\mathrel{\kl[\Jl]<_{\kl[\Jl]{\greenJ}}}}
\knowledge\Jl{notion}
\newrobustcmd\Jg{\mathrel{\kl[\Jg]>_{\kl[\Jg]{\greenJ}}}}
\knowledge\Jg{notion}

\newrobustcmd\greenL{\mathcal{L}}
\newrobustcmd\Leq{\mathrel{\kl[\Leq]{\greenL}}}
\knowledge\Leq{notion}
\newrobustcmd\Lleq{\mathrel{\kl[\Lleq]\leqslant_{\kl[\Lleq]{\greenL}}}}
\knowledge\Lleq{notion}
\newrobustcmd\Lgeq{\mathrel{\kl[\Lgeq]\geqslant_{\kl[\Lgeq]{\greenL}}}}
\knowledge\Lgeq{notion}
\newrobustcmd\Ll{\mathrel{\kl[\Ll]<_{\kl[\Ll]{\greenL}}}}
\knowledge\Ll{notion}
\newrobustcmd\Lg{\mathrel{\kl[\Lg]>_{\kl[\Lg]{\greenL}}}}
\knowledge\Lg{notion}

\newrobustcmd\greenR{\mathcal{R}}
\newrobustcmd\Req{\mathrel{\kl[\Req]{\greenR}}}
\knowledge\Req{notion}
\newrobustcmd\Rleq{\mathrel{\kl[\Rleq]\leqslant_{\kl[\Rleq]{\greenR}}}}
\knowledge\Rleq{notion}
\newrobustcmd\Rgeq{\mathrel{\kl[\Rleq]\geqslant_{\kl[\Rleq]{\greenR}}}}
\newrobustcmd\Rl{\mathrel{\kl[\Rl]<_{\kl[\Rl]{\greenR}}}}
\knowledge\Rl{notion}
\newrobustcmd\Rg{\mathrel{\kl[\Rl]>_{\kl[\Rl]{\greenR}}}}
\newrobustcmd\greenD{\mathcal{D}}
\newrobustcmd\Deq{\mathrel{\kl[\Deq]{\greenD}}}
\knowledge\Deq{notion}

\newrobustcmd\greenH{\mathcal{H}}
\newrobustcmd\Heq{\mathrel{\kl[\Heq]{\greenH}}}
\knowledge\Heq{notion}
\newrobustcmd\lettermap{\kl[\lettermap]{\sigma}}
\knowledge\lettermap{notion}
\newrobustcmd\ordmap{\kl[\ordmap]{\sigma}^{\kl[\ordmap]{\textnormal{ord}}}}
\knowledge\ordmap{notion}
\newrobustcmd\singordmap{\kl[\singordmap]{\tilde{\sigma}^{\textnormal{ord}}}}
\knowledge\singordmap{notion}

\newrobustcmd\exJ{\kl[\exJ]{J}}
\knowledge\exJ{notion}
\newrobustcmd\exK{\kl[\exK]{K}}
\knowledge\exK{notion}
\newrobustcmd\exL{\kl[\exL]{L}}
\knowledge\exL{notion}

\usepackage{marginnote}
\newcommand{\ie}{{\it i.e.}}

\title{First-order separation over countable ordinals%
\thanks{Paper accepted at \href{https://etaps.org/2022/fossacs}{FoSSaCS 2022}, and licensed under \href{http://creativecommons.org/licenses/by/4.0/}{CC BY 4.0}.\\ This work was supported by the European Research Council (ERC) under the 
European Union’s Horizon 2020 research and innovation programme (\href{https://cordis.europa.eu/project/id/670624}{ERC 
DuaLL, grant agreement No. 670624}), and by the DeLTA ANR project 
(\href{https://anr.fr/Projet-ANR-16-CE40-0007}{ANR-16-CE40-0007}).}}
\author[$\dag$]{\href{https://orcid.org/0000-0001-6529-6963}{Thomas Colcombet}}
\author[$\dag$]{\href{https://orcid.org/0000-0002-6360-6363}{Sam van Gool}}
\author[$\ddag$]{\href{https://orcid.org/0000-0002-1418-3405}{Rémi Morvan}}
\affil[$\dag$]{IRIF, Université de Paris \textit{\&} CNRS}
\affil[$\ddag$]{École Normale Supérieure Paris-Saclay}
\date{\today}

\newrobustcmd\remi[1]{\textcolor{green}{R: #1}}
\newrobustcmd\sam[1]{\textcolor{purple}{S: #1}}
\newrobustcmd\thomas[1]{\textcolor{red}{T: #1}}

\begin{document}

\maketitle

\begin{abstract}
  We show that the existence of a "first-order formula" "separating" two monadic second order formulas over "countable ordinal words" is decidable.
  This extends the work of Henckell and Almeida on "finite words", and of Place and Zeitoun on "$\omega$-words".
  For this, we develop the algebraic concept of "monoid" (resp. $\omega$-semigroup, resp. "ordinal monoid") with "aperiodic merge", an extension of "monoids" (resp. $\omega$-semigroup, resp. "ordinal monoid")  that explicitly includes a new operation capturing the loss of precision induced by first-order indistinguishability.
  We also show the computability of "\FO-pointlike sets", and the decidability of the "covering problem" for "first-order logic" on "countable ordinal words".

  \emph{Keywords:} regular languages, separation, pointlike sets,
  countable ordinals, first-order logic, monadic second-order logic.

  \alert{This document contains internal hyperlinks, and is best read on an electronic device.}
\end{abstract}

\section{Introduction}
\label{section:introduction}
% !TeX root = ../main-separation-ordinals.tex
% !TEX root =  ../main-separation-ordinals.tex

In this paper, we establish the decidability of "\FO-separability" over "countable ordinal words":
\begin{theorem}\AP\label{theorem:main}\phantomintro{separation problem}
	There is an algorithm which, given two "regular languages of countable ordinal words"~$K,L$, either:
	\begin{itemize}
	\item answers `""yes""', and outputs an ""\FO-separator"" which is an "\FO-formula"~$\varphi$
	which ""separates""~$K$ from~$L$, 
	"ie" such that~$u\models\varphi$ for all~$u\in K$, and~$v\models\neg\varphi$ for all~$v\in L$, or
	\item answers `""no""', and outputs a ""witness function"", i.e., a  computable function
	taking as input an "\FO-sentence" $\varphi$
	and returning a pair of words $(u,v) \in K \times L$
	such that $u \models \varphi$ if and only if
	$v \models \varphi$.
	\end{itemize}
\end{theorem}

%\paragraph*{Related work}
The decidability of "\FO-separability" was previously only known for finite words \cite{henckell1988,almeida1999some,zeitoun2016separating,gs2019merge}
and for words of length~$\omega$ \cite{zeitoun2016separating}.
"Countable ordinal words" are sequences of letters that are indexed by a countable total well-ordering, \ie, up to isomorphism, by a countable ordinal. There is a natural notion of
"regular languages@@COW" over these objects which can be equivalently described in terms of logic (either monadic second-order logic or weak monadic second-order logic), automata (Büchi introduced a notion of
automata for "countable ordinal words" \cite{buchi1973monadic},
which was studied in more detail by Wojciechowski
\cite{wojciechowski1984classes} and
which generalises Choueka's automata \cite{choueka1978finite}
for words of length at most \(\omega^n\)---the fact that Choueka's 
automata can be seen as a restriction of Büchi's automata
for "countable ordinals" was proven by Bedon \cite{bedon1996finite}), rational expressions (introduced by Wojciechowski \cite{wojciechowski1985finite}), or algebra (recognisable by finite "ordinal monoids"---introduced by Bedon and Carton
\cite{bedon1998eilenberg}). A detailed survey of the equivalence
between all these notions can be found in \cite{bedon1998these}.

Our algorithm follows the approach initiated by Henckell, and constructs the "\FO-pointlike sets" in an "ordinal monoid" that recognises the two input languages simultaneously.  %The explicit use of "witnesses of inseparability" is new, but is easily derivable from previous proofs. 
"\FO-pointlike sets" are subsets of a monoid whose elements are inherently indistinguishable by first-order logic. Our completeness proof for the algorithm follows a scheme similar to the one followed by Place and Zeitoun in the context of finite and $\omega$-words \cite{zeitoun2016separating}, which was inspired by Wilke's characterisation of "\FO-definable@@lang" languages \cite{wilke1999classifying}. We had to make several substantial changes to this approach for the proofs to generalize from finite and $\omega$-words to the setting of "countable ordinal words". A seemingly slight modification of the notion of "saturation" (\Cref{definition:sat}) allows for a careful redesign of several of the core lemmas in the proof of completeness, and in particular the construction of an "\FO-approximant" in \Cref{section:yes} below.

\paragraph*{Related work}

This work lies in a line of research that aims to obtain a decidable  understanding of the expressive power of subclasses of the class of regular languages.
The seminal work in this area is the Schützenberger-McNaughton-Papert theorem \cite{schutzenberger1965finite,mcnaughton1971counter} which effectively characterizes the languages of finite words definable in first-order logic as the ones which have an "aperiodic" "syntactic monoid". This theorem was at the origin of a large body of work that studies classes of languages through the corresponding classes of monoids, including for instance Simon's result characterising piecewise-testable languages via "\(\Jeq\)-trivial" monoids \cite{simon1975piecewise}.
"\FO-pointlike sets" are also known in the literature as "aperiodic pointlike sets", and were first studied and shown to be computable by Henckell \cite{henckell1988}, in the context of the Krohn-Rhodes semigroup complexity problem. The computability of pointlike sets was shown to be equivalent to the decidability of the "covering problem" by Almeida \cite{almeida1999some}. Alternative proofs of separation and covering problems for {\FO} were given recently in \cite{zeitoun2016separating,gs2019merge}, and, ever since Henckell's work, the computability of "\FO-pointlike sets" was also extended to pointlike sets for other varieties---for example
\cite{ash1991inevitable} for the variety of finite groups, 
\cite{az1997pseudovariety} for the variety of \(\Jeq\)-trivial finite semigroups
and \cite{gs2019groups} for varieties of finite semigroups determined by 
a variety of finite groups; also see \cite{gs2019groups} for further references. Place and Zeitoun recently used pointlike sets, in the form of "covering problems" \cite{pz2018covering}, to resolve long-standing open membership problems for the lower levels of the dot-depth and of the Straubing-Thérien hierarchies \cite{place2018complexity,pz2019all,pz2021separation}.

Another, orthogonal, line of research consists in the extension of the notions of regularity (logic/automata/rational expressions/algebra) to models beyond finite words. This is the case for  finite or infinite trees \cite{rabin69}. In this paper, we are concerned with words that go beyond finite, such as words of length~$\omega$ \cite{buchi62,wilke1993algebraic,perrinpinpin}, of countable ordinal length \cite{bedon1998these,bedon1996finite}, of countable scattered\footnote{\AP A linear ordering is ""scattered@@linord"" if it does not contain a dense subordering.} length \cite{rispal2004thesis,rispal2005complementation}, or of general countable length \cite{rabin69,shelah75,carton2018algebraic}. 

These two branches have also been studied jointly, and first-order logic was characterised on words of length~$\omega$ \cite{perrin1984recent}, of countable ordinal length \cite{bedon2001logic},
of countable scattered length \cite{bes2011algebraic} (and in \cite{bedon2012regular} for first-order augmented with quantifiers over Dedekind cuts), and for words of countable length \cite{colcombet2015limited} (as well as other logics \cite{colcombet2015limited,ManuelSreejith16,adsul2021firstorder}). Prior to the current work, the questions of computing the "\FO-pointlike sets" and deciding "\FO-separation" for languages of infinite words had only been investigated for words of length~$\omega$  \cite{zeitoun2016separating}.

\paragraph*{Structure of the document}
In \Cref{section:preliminaries}, we introduce important definitions for manipulating infinite words in algebraic terms ("ordinal monoids" and their powerset), and in logical terms ("first-order logic" and "first-order definable maps").
In \Cref{section:algorithm}, we describe the "algorithm", and in particular its core, a "saturation" construction.
The correctness of the algorithm is then proved in \Cref{section:no}, and the completeness in \Cref{section:yes}.
In \Cref{section:related}, we show two stronger results that arise from the same technique: the decidability of the "covering problem" and the computability of "pointlikes".
\Cref{section:conclusion} concludes.

\section{Preliminaries}
\label{section:preliminaries}
% !TeX root = ../main-separation-ordinals.tex
% !TEX root = ../main-separation-ordinals.tex

\subsection{Ordinals}
% \begin{itemize}
% \item ""linear ordering"", ""countable ordinals"",
% \item ""countable word"", ""countable ordinal word"", ""concatenation"", ""omega iteration"", $\intro*\Womega$, $\intro*\ord$, $\intro*\ordp$
% \item ""condensation""
% \end{itemize}
\AP A ""linear ordering"" is a set equipped with a total order.
It is ""countable@@linord"" (resp. ""finite@@linord"") if the underlying set is countable (resp. finite).
Let $\alpha$ and $\beta$ be two "linear orderings".
A ""morphism@@linord"" from $\alpha$ to $\beta$ is a monotonic function,
\AP
and an ""isomorphism@@linord"" between $\alpha$ and $\beta$
is a bijective "morphism@@linord". 
\AP
The (ordered) ""sum@@linord"" of two linear orders $\alpha$ and $\beta$ is denoted by $\alpha \intro*\pluslinord \beta$ and is defined, as usual, on the disjoint union of the linear orders $\alpha$ and $\beta$, by further postulating that every element of $\alpha$ is below every element of $\beta$. The ""product@@linord"" of two linear orders is denoted by $\alpha \intro*\timeslinord \beta$ and is defined to be the right-to-left lexicographic ordering on the Cartesian product of the two orders, "ie", $(x,y) \leqslant (x', y')$ iff $y < y'$ or $y = y'$ and $x \leqslant x'$. The $n$-fold product of $\alpha$ with itself is denoted by $\alpha^n$.
%\AP An ""embedding@@linord"" is an injective "morphism@@linord".
% \AP We define the ""sum@@linord"", denoted by $\pluslinord$,
% and the ""product@@linord"", denoted by $\timeslinord$,
% of two "linear orderings" $(\alpha,\leqslant_\alpha)$ and
% $(\beta,\leqslant_\beta)$ as the following "linear orderings":\thomas{Pas sur que ce soit la bonne façon d'introduire ces notions.}
% \begin{center}
% \begin{tabular}{ccc}
% 	\toprule "linear ordering" & carrier & order \\ \midrule
% 	\multirow{2}{*}{$(\alpha,\leqslant_\alpha) \pluslinord
% 	(\beta,\leqslant_\beta)$}
% 		& \multirow{2}{*}{disjoint union $\alpha + \beta$}
% 		& $x \leqslant_{\alpha\pluslinord \beta} y$ iff
% 			[$x,y \in \alpha$ and $x \leqslant_{\alpha} y$] \\
% 		& & or [$x,y \in \beta$ and $x \leqslant_{\beta} y$]
% 			or [$x\in \alpha$ and $y \in \beta$] \\
% 	\multirow{2}{*}{$(\alpha,\leqslant_\alpha) \timeslinord
% 	(\beta,\leqslant_\beta)$}
% 		& \multirow{2}{*}{cartesian product $\alpha \times \beta$}
% 		&  $(x,y) \leqslant_{\alpha\timeslinord \beta} (x',y')$ iff \\
% 	& & $y < y'$ or [$y = y'$ and $x \leqslant x'$] \\ \bottomrule
% \end{tabular}
% \end{center}
% 
\AP A "linear ordering" is ""well-founded""
when it does not contain an infinite strictly decreasing sequence.
\AP An ""ordinal"" is a "well-founded" "linear ordering", considered only up to
"isomorphism@@linord" of linear orderings.
\AP The empty "linear ordering",
the "linear ordering" with a single element and
the "linear ordering" of natural numbers are all "ordinals",
and are denoted $0$, $1$ and $\omega$, respectively.
The class of all "ordinals" is itself totally ordered by the ""embedding@@linord""
relation: $\alpha \intro*\leqord \beta$ means that there exists an injective monotonic function from $\alpha$ to $\beta$.
The relation $\intro*\lessord$ denotes the strict ordering associated
with $\leqord$. \AP An "ordinal" is a ""successor ordinal""  if it has a maximum,
and a ""limit ordinal"" otherwise.

% \debarras{\thomas{to process for usefulness for the rest of this paragraph}
% It is routine to check that a "limit ordinal" $\lambda$ is the least upper bound, under $\leqord$,
% of all ordinals strictly smaller than $\lambda$, "ie"
% $\lambda = \intro*\lub_{\alpha\lessord\lambda} \alpha$.
% The ""exponentiation@ord""
% $\intro*\expord{\alpha}{\beta}$ of an ordinal $\alpha$
% by an ordinal $\beta$ is defined by transfinite induction over $\beta$,
% by $\expord{\alpha}{\beta_0 + 1} = \alpha^{\beta} \timesord \alpha$ if
% $\beta = \beta_0 + 1$ is a "successor ordinal",
% and $\expord{\alpha}{\beta} = \lub_{\gamma<\beta} \alpha^\gamma$.
% For example, $\expord{2}{\omega} = \lub_{n<\omega} 2^n = \omega$.}
% If $\alpha$ is an ordinal and $n\in \Nats$, we use the notation 
% \begin{align*}
% 	\alpha^n :=
% 	\underbrace{\alpha \timesord \hdots \timesord \alpha}_{n \text{ times}}.
% \end{align*}
%We will use the following proposition, which is a corollary of
%the existence and uniqueness of ``Cantor's normal form'', for which see, e.g., \cite[Thm.~2.26]{jech2006set}. \sam{Remove the %following and just place the reference directly in the proof of Lemma 24 where it is used.}
%
%\begin{proposition}
%	\label{prop:cantor-normal-form}
%	For every $\ell\in\Nats$, every countable ordinal $\iota$
%	can be uniquely written as
%	\begin{align*}
%		\iota = \omega^\ell \timesord \iota_\ell \plusord
%		\omega^{\ell-1}\timesord i_{\ell-1} \plusord
%		\hdots \plusord
%		\omega^1 \timesord i_1 \plusord
%		i_0
%	\end{align*}
%	where $i_0,\hdots,i_{\ell-1} \in \Nats$ and $\iota_\ell$
%	is a countable ordinal.
%\end{proposition}

\subsection{Ordinal words}

\AP Given a set $X$, a ""word""~$w$ over $X$ is a map from some "linear ordering" to $X$.
The "linear ordering" is called the ""domain"" of~$w$, and denoted $\intro*\dom (w)$.
\AP
A "word" is ""countable@@word"" (resp. ""finite@@word"", resp. ""scattered@@word"", resp. ""$\omega$-word""),
if its "domain" is~"countable@@linord" (resp. "finite@@linord", resp. "scattered@@linord", resp. $\omega$).
\AP In this paper, a ""countable ordinal word"" is a "word" that has a "countable@@linord" and "ordinal@@linord" domain
(hence, the countability assumption in silently assumed throughout the paper).
%\AP \thomas{We have to clarify the 'up to isomorphism'.}
\AP The set of all "finite words" over $X$ is 
denoted by $X^\ast$, and 
the collection of all "countable ordinal words" over $X$ is denoted by $X\ord$. Similarly, the set of finite non-empty words is denoted by $X^+$ and the collection of non-empty countable ordinal words is denoted by 
$X\intro*\ordp$.
\AP The concatenation of two "countable ordinal words" $u$ and $v$
over $X$ is the word $u\cdot v: \dom(u) \plusord \dom(v) \to X$
over $X$ defined by $(u\cdot v)_\iota := u_\iota$
if $\iota \in \dom(u)$ and $(u\cdot v)_\iota := v_\iota$
if $\iota \in \dom(v)$. 
\AP If $w$ is a "countable ordinal word",
we define its ""omega iteration"", denoted by $w\intro*\Womega$, as
the word with domain $\dom(w)\timesord\omega$ defined by $(w^\omega)_{(\iota, n)} := w_{\iota}$ for every $\iota \in \dom(w)$ and $n \in \omega$.
% and consisting of $\omega$ copies of the word $w$.
For example, if $a,b\in X$, then the "omega iteration" $(ab)\Womega$
of the two-letter word $ab$ is the word $ababab\cdots$ with domain $2 \timesord \omega = \omega$.

%\subsection{Algebra}
%\label{subsection:preliminaries-algebra}

% \thomas{find a way to make these definition compact. In particular, I am not sure that it is necessary to introduce ordinal algebras. Just ordinal monoids should be sufficient.}
% \begin{itemize}
% \item ""semigroups"", ""monoids"", $\intro*\unit$,
% \item ""idempotent power"", $-\intro*\idem$
% \item $-\intro*\Aomega$
% \item ""ordinal monoid"", ""ordinal monoid morphism"",
% \item description of the powerset of these constructions.  $-\Pomega$
% \end{itemize}

\subsection{Ordinal monoids}

\AP A ""semigroup"" is a set $S$ equipped with an associative binary product, denoted by~$\cdot$. 
A ""monoid"" is a "semigroup" with a
distinguished neutral element for the product, denoted as $\intro*\unit$.
\AP %In a semigroup $S$, a power of an element $x \in S$ is an element of the form $x^k$ for some $k>0$.
An element $x \in S$ is called ""idempotent"" if $x^2 = x$. %\debarras{---for example, the unit of a "monoid" is idempotent}.
\AP In a finite finite semigroup~$S$, every element $x \in S$ has a unique ""idempotent
power"", denoted by\footnote{The standard notation is~$x^\omega$, but this notation conflicts with the linear ordering~$\omega$.
	It is sometimes denoted $x^\pi$ or $x^!$ when in the context
	of infinite words. We find the notation~$x\idem$ more self-explanatory.}
  $x\intro*\idem$, which we recall is the limit of the ultimately constant series $n\mapsto x^{n!}$. %(for a proof, see e.g. \cite[prop II.6.31]{pin2020mpri}).
 We also denote~$x\reintro*\idem[+k]$, for $k$~integer,
the limit of the ultimately constant series $n\mapsto x^{n!+k}$. 
%Note that the set~$\{x\idem[+k] \colon k \geq 0\}$ contains exactly the elements
Note that $x\idem$ is the identity element of the unique maximal group inside the subsemigroup generated by~$x$.
\AP A finite semigroup is ""aperiodic"" (we equivalently write ""group-trivial"") if~$a\idem=a\idem[+1]$ for all of its elements~$a$.

\AP We now extend the notion of "monoid" to obtain an algebraic structure
in which one can evaluate a product indexed by any "countable ordinal". Let $\alphabet$ be any set, and $\alpha$ a "countable ordinal". 
\AP For any "word" $(w_\iota)_{\iota < \alpha}$ over the set $\alphabet\ord$ of "countable ordinal words"---i.e. $(w_\iota)_{\iota < \alpha}$ is a "word"
whose letters are "words" over $\alphabet$---
we define $\intro*\flatten(w_\iota \mid \iota < \alpha)$ to be the "word" over $\alphabet$
with "domain" $\sum_{\iota < \alpha} \dom(w_\iota)$,
which has the letter $(w_\iota)_{\kappa} \in \alphabet$ at position $(\iota,\kappa)$, for every $\iota \in \alpha$ and $\kappa \in \dom(w_\iota)$.
%\remi{Phrase suivante à reformuler…}
%The definition relies on the notion of ""flattening"":
%a word $w\in (\alphabet\ord)\ord$ whose letters are words over $\alphabet$
%can be ``flattened'' to obtain a word $\intro*\flatten(w) \in \alphabet\ord$.
%Dually, one can see $w$ as a way of parenthesising $\flatten(w)$.
%\thomas{What about a notation $\flatten(u_i\mid i\in\alpha)$, ou même `concat(...)'.}

\begin{definition}\AP
	An ""ordinal monoid""\footnote{The object should probably be called a `countable ordinal monoid' since its intent is to model countable ordinal words. However the naming becomes clumsy for `finite countable ordinal monoids'...}
is a pair $\intro*\ordmonoid = (M,\product)$
where $M$ is a set and
$\intro*\product: M\ord \to M$ is a function,
called ""generalised product"", such that: 
\begin{itemize} 
	\item $\product(x) = x$ for every $x \in M$, and
	\item $\product((\product(u_\iota))_{\iota<\alpha})
	= \product(\flatten((u_\iota)_{\iota<\alpha}))$
	for every word $(u_\iota)_{\iota<\alpha} \in (M\ord)\ord$.
\end{itemize}
\end{definition}

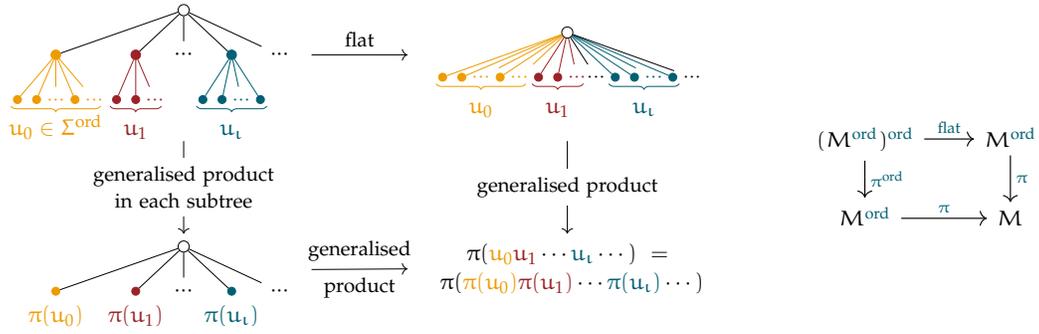
\begin{figure}%
	\begin{center}%
	\begin{tabular}{cc}%
		\hspace{-.8cm}
		\begin{tabular}{c}%
			\scalebox{.85}{\tikzset{
    c/.style = {align=center, inner sep=1.75pt, draw=black, circle},
    ca/.style = {
        align=center, inner sep=1.25pt, draw=colora, fill=colora, circle
    },
    cb/.style = {
        align=center, inner sep=1.25pt, draw=colorb, fill=colorb, circle
    },
    cc/.style = {
        align=center, inner sep=1.25pt, draw=colorc, fill=colorc, circle
    },
}
\newrobustcmd\smallcdots{{\cdot}{\cdot}{\cdot}}
\definecolor{colora}{HTML}{ee9b00}
\definecolor{colorb}{HTML}{9b2226}
\definecolor{colorc}{HTML}{005f73}
\begin{tikzpicture}
    \begin{scope}
    \tikzset{
        level 1/.style = {sibling distance = 1cm},
        level 2/.style = {sibling distance = .3cm},
        level distance = .7cm
    }
    \node[c, fill=white] at (0,0) {}
        child { node[ca, inner sep=1.5pt] {}
            child { node[ca] (al) {} edge from parent[colora] }
            child { node[ca] {} edge from parent[colora] }
            child {
                node[font=\scriptsize, colora] {$\smallcdots$}
                edge from parent[colora]
            }
            child { node[ca] {} edge from parent[colora] }
            child {
                node[font=\scriptsize, colora] (ar) {$\smallcdots$}
                edge from parent[colora]
            }
        }
        child[sibling distance = .75cm] { node[cb, inner sep=1.5pt] {}
            child { node[cb] (bl) {} edge from parent[colorb] }
            child { node[cb] {} edge from parent[colorb] }
            child {
                node[font=\scriptsize, colorb] (br) {$\smallcdots$}
                edge from parent[colorb]
            }
        }
        child[sibling distance = .75cm] {
            node[font=\small] {$\smallcdots$}
        }
        child[sibling distance = .75cm] { node[cc, inner sep=1.5pt] {}
            child { node[cc] (cl) {} edge from parent[colorc] }
            child { node[cc] {} edge from parent[colorc] }
            child {
                node[font=\scriptsize, colorc] {$\smallcdots$}
                edge from parent[colorc]
            }
            child { node[cc] (cr) {} edge from parent[colorc] }
        }
        child[sibling distance = .75cm] {
            node[font=\small] (right-tree-ul) {$\smallcdots$}
        };
        \draw[decorate, decoration={brace, mirror}, colora] ($(al)+(-.1,-.1)$) -- ($(ar)+(.1,-.1)$) 
        node[font=\small, below, midway, colora] {
            $u_0 \in \kl[\alphabet]{\color{colora}\Sigma}^{\kl[\ord]{\color{colora}\mathrm{ord}}}$
        };
        \draw[decorate, decoration={brace, mirror}, colorb] ($(bl)+(-.1,-.1)$) -- ($(br)+(.1,-.1)$) 
        node[font=\small, below, midway, colorb] {$u_1 \vphantom{\alphabet\ord}$};
        \draw[decorate, decoration={brace, mirror}, colorc] ($(cl)+(-.1,-.1)$) -- ($(cr)+(.1,-.1)$) 
        node[font=\small, below, midway, colorc] {$u_\iota \vphantom{\alphabet\ord}$};
    \end{scope}

    \begin{scope}
    \tikzset{
        level distance = .7cm,
        level 1/.style = {sibling distance = .3cm}
    }
    \node[c] at (6,-0.35) {}
        child { node[ca] (a2l) {} edge from parent[colora] }
        child { node[ca] {} edge from parent[colora] }
        child { node[font=\scriptsize, colora] {$\smallcdots$} edge from parent[colora] }
        child { node[ca] {} edge from parent[colora] }
        child { node[font=\scriptsize, colora] (a2r) {$\smallcdots$} edge from parent[colora] }
        child { node[cb] (b2l) {} edge from parent[colorb] }
        child { node[cb] {} edge from parent[colorb] }
        child { node[font=\scriptsize, colorb] (b2r) {$\smallcdots$} edge from parent[colorb] }
        child { node[font=\scriptsize] {$\smallcdots$} }
        child { node[cc] (c2l) {} edge from parent[colorc] }
        child { node[cc] {} edge from parent[colorc] }
        child { node[font=\scriptsize, colorc] {$\smallcdots$} edge from parent[colorc] }
        child { node[cc] (c2r) {} edge from parent[colorc] }
        child { node[font=\scriptsize] {$\smallcdots$} };
        \draw[decorate, decoration={brace, mirror}, colora] ($(a2l)+(-.1,-.1)$) -- ($(a2r)+(.1,-.1)$) 
        node[font=\small, below, midway, colora] {$u_0 \vphantom{\alphabet\ord}$};
        \draw[decorate, decoration={brace, mirror}, colorb] ($(b2l)+(-.1,-.1)$) -- ($(b2r)+(.1,-.1)$) 
        node[font=\small, below, midway, colorb] {$u_1 \vphantom{\alphabet\ord}$};
        \draw[decorate, decoration={brace, mirror}, colorc] ($(c2l)+(-.1,-.1)$) -- ($(c2r)+(.1,-.1)$) 
        node[font=\small, below, midway, colorc] {$u_\iota \vphantom{\alphabet\ord}$};
    \end{scope}

    \begin{scope}
        \tikzset{
            level distance = .7cm,
            level 1/.style = {sibling distance = 1cm}
        }
        \node[c] at (0,-3.7) {}
            child { node[ca] (a3) {} }
            child[sibling distance = .75cm] { node[cb] (b3) {} }
            child[sibling distance = .75cm] { node[font=\small] {$\smallcdots$} }
            child[sibling distance = .75cm] { node[cc] (c3) {} }
            child[sibling distance = .75cm] { node[font=\small] {$\smallcdots$} };
        \node[font=\small, colora, below = -.5mm of a3] {
            $\kl[\product]{\color{colora}\pi}(u_0) \vphantom{\alphabet\ord}$
        };
        \node[font=\small, colorb, below = -.5mm of b3] {
            $\kl[\product]{\color{colorb}\pi}(u_1) \vphantom{\alphabet\ord}$
        };
        \node[font=\small, colorc, below = -.5mm of c3] {
            $\kl[\product]{\color{colorc}\pi}(u_\iota) \vphantom{\alphabet\ord}$
        };
    \end{scope}

    \draw (2,-.7) edge[->] node[midway, above, font=\footnotesize] {\kl[\flatten]{\color{black}flat}} (3.5,-.7);
    \draw (0,-2.05) edge[->] node[midway, fill=white, font=\footnotesize, text width=3cm,align=center] {\kl[generalised product]{\color{black}generalised product} in each subtree} (0,-3.5);
    \draw (6,-2.05) edge[->] node[midway, fill=white, font=\footnotesize, text width=3cm,align=center] {\kl[generalised product]{\color{black}generalised product}} (6,-3.5);
    \draw (2,-4.06) edge[->] node[midway, above, font=\footnotesize] {\kl[generalised product]{\color{black}generalised}} node[midway, below, font=\footnotesize] {\kl[generalised product]{\color{black}product}} (3.5,-4.05);

    \node[font=\small, text width=5cm, align=center] at (6,-4.05) {
        $\kl[\product]{\color{black}\pi}(\textcolor{colora}{u_0} \textcolor{colorb}{u_1} \cdots \textcolor{colorc}{u_\iota} \cdots) 
        = \kl[\product]{\color{black}\pi}(\textcolor{colora}{\kl[\product]{\color{colora}\pi}(u_0)} \textcolor{colorb}{\kl[\product]{\color{colorb}\pi}(u_1)} \cdots \textcolor{colorc}{\kl[\product]{\color{colorc}\pi}(u_\iota)} \cdots)$
    };
\end{tikzpicture}}
		\end{tabular}
		&
		\begin{tabular}{c}\scriptsize
		\begin{tikzcd}
			(M\ord)\ord \rar["\flatten"] \dar["\product\ord"] & M\ord \dar["\product"] \\
			M\ord \rar["\product"] & M
		\end{tikzcd}
		\end{tabular}
	\end{tabular}
	\end{center}
	\caption{\label{fig:generalised-associativity}"Generalised associativity", pictorially (left) and diagrammatically (right).}
\end{figure}

\AP The second axiom is called ""generalised associativity"",
and is illustrated in \Cref{fig:generalised-associativity}. %: given a word, one can either compute its "generalised product", or first arbitrarily parenthesise the word, evaluate every subword, and then evaluate the result: both method must yield the same result.
% The "generalised product" of the empty word
% $\product(\varepsilon)$ is denoted by 1, and is an identity
% in $\ordmonoid$.
%\AP ""Ordinal semigroups""
%are defined similarly to "ordinal monoids", except that we replace
%every occurrence of $-\ord$ by $-\ordp$.
%Note that every "ordinal monoid" is an "ordinal semigroup".
\AP
An ""ordinal monoid morphism"" is a map between
"ordinal monoids" preserving the "generalised product".
% Moreover, an ""ordinal submonoid""  of an "ordinal monoid" 
% is a subset of the "ordinal monoid"
% stable under "generalised product". 
%\AP As expected, a ""morphism of ordinal monoid"" is a map which is compatible with the product.
\AP An "ordinal monoid" is ""ordered@ordered ordinal monoid"" if it is equipped with
an order~$\leqslant$ that makes~$\product$ monotonic, "ie" such that~$u\leqslant v$ implies~$\product(u)\leqslant\product(v)$,
in which $\leqslant$ is extended letter-by-letter to words in~$M\ord$.

\AP Given a set~$\intro*\alphabet$ (the ""alphabet""), an "ordinal monoid"~$\monoid=(M,\product)$, a letter-to-letter map $\intro*\lettermap\colon\alphabet\to M$ extended to $\intro*\ordmap\colon\alphabet\ord\to M\ord$, and~$F\subseteq M$, the language~$L\subseteq\alphabet\ord$ ""recognised@@OM"" by  $(\monoid,\lettermap,F)$ is
\begin{align*}
	L = \{u \in \alphabet\ord \ \colon \ \product(\ordmap(u))\in F \},
\end{align*}
and a language~$L\subseteq\alphabet\ord$ is called \reintro*"recognisable@@OM" if it is recognised by some such tuple $(\monoid,\lettermap,F)$.
\AP We recall that "recognisable@@OM" languages of ordinal words coincide with the ones definable in monadic second-order logic, or definable by suitable automata. 
These languages are called ""regular@@cow"". \Cref{example:saturation} below will illustrate this concept.

%Moreover, an ""ordinal subsemigroup"" (resp. ""ordinal submonoid"")
%of an "ordinal semigroup" (resp. "monoid@ordinal monoid")
%is a subset of the "ordinal semigroup" (resp. "monoid@ordinal monoid")
%stable under "generalised product".
% \AP \thomas{On a besoin des congruences ?} \remi{I think we need
% them on words, at least to say that the FO-congruence is a congruence…
% On peut ne pas parler de "congruence" si vous voulez et établir les
% propriétés à la main.} Finally,
% an ""ordinal semigroup congruence"" over an "ordinal semigroup" $S$
% is an equivalence relation over $S$ that is also an
% "ordinal subsemigroup" of $S\times S$.

\AP We now recall a finite presentation of finite "ordinal monoids" (originally for ordinal semigroups), first given by
Bedon \cite{bedon1998these} by extending a similar
result established by Perrin and Pin \cite[prop II.5.2]{perrinpinpin}
for $\omega$-semigroups\footnote{The finitary reprensation of
$\omega$-semigroups is usually called a Wilke algebra, which is the
algebraic structure introduced by Wilke in \cite{wilke1993algebraic} to
recognise regular $\omega$-languages.}.
Let $(S,\product)$ be an
"ordinal monoid".
\AP
We define the constant~$\intro*\Aunit$ and two functions
	$\intro*\Acdot: S\times S \to S$ and $\intro*-\Aomega: S \to S$ by
%\thomas{À valider les notations $\Aomega$ et $\Womega$...}
\begin{align*}
	\reintro*\Aunit&:=\product(\varepsilon)&
	x \reintro*\Acdot y &:= \product(xy) &
	\text{and}\qquad  x\reintro*\Aomega &:= \product(x\Womega) = \product(\overbrace{xxx\cdots}^{\text{$\omega$ times}})\ .
\end{align*}
\AP The following proposition lets us interchangeably regard an "ordinal monoid"
$\ordmonoid$ as either a pair $(M,\product)$ or as a quadruple
$(M,\Aunit,\Acdot,-\Aomega)$, that we refer to as its ""presentation@@OM"".
\begin{proposition}[{\cite[Thm.~3.5.6]{bedon1998these}}, originally for ordinal semigroups]
	In a finite "ordinal monoid"
	the "generalised product" is uniquely determined 
	by the operations $\Aunit,\Acdot$ and $-\Aomega$.
\end{proposition}

\AP An important construction on which our proof relies is the
""power ordinal monoid"": given an "ordinal monoid" 
$(M,\product)$, we equip the powerset $\intro*\Pset(M)$ of $M$
with a "generalised product" $\intro*\Pproduct: \Pset(M)\ord
\to \Pset(M)$ defined by
\begin{multline*}
	\Pproduct((X_\iota)_{\iota<\kappa}) :=
	\left\{
		\product((x_\iota)_{\iota<\kappa}) \mid
		x_{\iota} \in X_\iota \text{ for all } \iota < \kappa
	\right\}\\
	\text{for all words $(X_\iota)_{\iota<\kappa} \in (\Pset(M))\ord$.}
\end{multline*}

Observe that if $M$ is a finite "ordinal monoid", then so is $\Pset(M)$.
We can compute a finite representation
of the "power ordinal monoid" $\Pset(M)$ of $M$
from a finite representation of $M$.
\AP Indeed,  \phantomintro\Pcdot\phantomintro\Pomega\phantomintro\Punit
\begin{align*}
	\reintro*\Punit &= \{\Aunit\},&
	X \reintro*\Pcdot Y &= \{x \cdot y \mid x\in X,\, y \in Y\}, &\text{and}\qquad
	X\reintro*\Pomega &= \{u \cdot v\Aomega \mid u,v \in X^+\}
\end{align*}
for all $X,Y \in \Pset(M)$.
The two first properties are trivial while the third one
can be proven using the infinite Ramsey's theorem---this is a classical 
argument used to give finite representation of infinite structures, see e.g. \cite[Theorem II.2.1]{perrinpinpin}.
%\sam{which theorem does this refer to?} \thomas{Sans importantce: c'est super classique.}
\AP Note that this "power ordinal monoid" is indeed an "ordinal monoid". It is even an "ordered ordinal monoid" when equipped with the inclusion ordering.

\subsection{First-order logic}
%\label{subsection:preliminaries-logic}
% \begin{itemize}
% \item ""first-order logic"" (\reintro{\FO-logic} for short), ""\FO-formula"", ""\FO-sentence"", $\models$
% \item ""\FO-definable condensations""
% \item ""\FO-definable maps""
% \end{itemize}
% 
%\AP In the rest of the paper, we fix a infinite set of ""variables""
%$\intro*\Var$.
\phantomintro{first-order logic}%
\AP Over a fixed (finite) alphabet $\alphabet$, we define the
set of ""first-order logic formul\ae{}@\FO-formula"" or
\reintro*"\FO-formul\ae{}@\FO-formula" for short, by the grammar:
\begin{align*}
	\varphi ::= \ \exists x.\,\varphi
	\ \ \mid\ \ \forall x.\,\varphi
	\ \ \mid\ \ \varphi \land \varphi
	\ \ \mid\ \ \varphi \lor \varphi
	\ \ \mid\ \ \neg \varphi
	\ \ \mid\ \ x < y
	\ \ \mid\ \ a(x)
\end{align*}
where $x,y$ range over some fixed infinite set of variables, and $a$ over
$\alphabet$. ""Free variables"" are defined as usual,
and an ""\FO-sentence"" is a formula with no "free variables".
\AP
In our setting, a model is a
"countable ordinal word", and a ""valuation"" over this model is a total map
from variables to the domain of the "word@@ord".
\AP
We define, for any word $w$ and any valuation $\nu$, the semantic relation $w, \nu \models \varphi$
of "first-order logic" on "countable ordinal words"
by structural induction on the "\FO-formula" $\varphi$,
by interpreting variables
as positions in the "word@@ord" and propositions of the form
$a(x)$ as ``the letter at position $x$ is an $a$''.
If $\varphi$ is an "\FO-sentence", then the semantics of
$\varphi$ over a word $w$ does not depend
on the valuation, and thus we write
$w \models \varphi$ or $w \not\models \varphi$. When $w \models \varphi$ we say that $w$ ""satisfies"" $\varphi$, or also that $\varphi$ ""accepts"" $w$.

\AP A language $L \subseteq \alphabet\ord$ is said to be
""\FO-definable@@lang"" if $L=\{w\in \alphabet\ord\mid w \models \varphi\}$ for some "\FO-sentence"
$\varphi$.
For example, the language of "words@@ord" over the alphabet $\{a,b,c\}$
such that every `$a$' is at a finite distance from a `$b$' is
defined by the "\FO-sentence" $\forall x. a(x) \rightarrow \exists y. b(y)\land \finite(x,y)$, where:%
\phantomintro\isSuccessor
%\phantomintro\isLimit
\phantomintro\finite
\begin{align*}
	\reintro*\isSuccessor(z) &::= \exists y. y<z\wedge
		\left(\forall x.\, x<z \rightarrow x\leqslant y\right)\\
	%\reintro*\isLimit(z) &::= \neg \isSuccessor(z) \\
	\reintro*\finite(x,y) &::= \forall z. (x< z\leqslant y \vee y< z\leqslant x )\rightarrow \isSuccessor(z)\ .
\end{align*}

Bedon \cite{bedon2001logic} extended the 
Schützenberger-McNaughton-Papert theorem
\cite{schutzenberger1965finite,mcnaughton1971counter} to
"countable ordinal words".

\begin{proposition}[{""Bedon's theorem"" \cite[Theorem 3.4]{bedon2001logic}}]
	\label{prop:bedon-thm}
	\AP A language of "countable ordinal words" is "\FO-definable@@lang"
	if and only if it is "recognised@@OM" by
	a finite "aperiodic" "ordinal monoid".
\end{proposition}

\AP Let $L \subseteq \alphabet\ord$.
A function $f: L \to X$ whose codomain $X$
is a finite set is said to be ""\FO-definable@@map"" when
every preimage $f^{-1}[x]$, with $x\in X$, is an
"\FO-definable language".
Note that if $f$ is
"\FO-definable@@map", then its domain $L$
is necessarily an "\FO-definable language".
% Equivalently, one
% can see "\FO-definable functions" as the partial functions
% $f: \alphabet\ordp \pto X$ such that there exists $k\in\Nats$,
% there exist $x_1,\hdots,x_k \in X$, "\FO-sentences"
% $\varphi_1,\hdots,\varphi_k$ such that
% the $\varphi_i$'s ($i\leqslant k$) are semantically disjoint ---
% if for some word $w\in\alphabet\ord$ we have $w \models \varphi_i$
% and $w \models \varphi_j$ then $i=j$---and $f$ is defined by:
% \begin{align*}
% 	f(w) := \begin{cases}
% 		x_1 & \text{ if $w \models \varphi_1$,} \\
% 		\;\vdots & \hspace{.9cm}\vdots\\
% 		x_k &  \text{ if $w \models \varphi_k$,} \\
% 		\texttt{undefined} & \text{ otherwise}
% 	\end{cases}
% \end{align*}
% for every word $w\in\alphabet\ord$.

\AP For example, the function $\alphabet^*\to\Zed/2\Zed$,
sending a word $w \in \alphabet^*$ to its length modulo 2,
is not "\FO-definable@@map". On the other hand, for a  fixed
letter $a\in \alphabet$, the total function sending a word $w \in
\alphabet\ordp$ to $\top$ if $w$ contains the letter `$a$' and
to $\bot$ otherwise is "\FO-definable@@map".

% \begin{definition}[If-then-else]
% 	\label{def:if-then-else-fo-definable}
% 	If $f: \alphabet\ord \pto X$ and $g: \alphabet\ord \pto X$
% 	are "\FO-definable functions" whose support are disjoint,
% 	then the function
% 	\begin{center}\begin{tabular}{cccc}
% 		$f \intro*\cupfun g \colon$ & $\alphabet\ord$ & $\pto$ & $X$ \\
% 		& $w$ & $\mapsto$ & $\begin{cases}
% 			f(w) & \text{ if $f(w)$ is defined,} \\
% 			g(w) & \text{ if $g(w)$ is defined,} \\
% 			\mathtt{undefined} & \text{ otherwise}
% 		\end{cases}$
% 	\end{tabular}\end{center}
% 	is also "\FO-definable@@map".
% \end{definition}

\AP A useful tool to manipulate words is the notion of
"condensation"---see, e.g.,
\cite[\S 4]{rosenstein1982linear} for an introduction to the subject.
A ""condensation"" of a
"countable ordinal" $\alpha$ is an equivalence
relation $\intro*\condensation$ over $\alpha$ whose equivalence
classes are convex. Note that the quotient of an "countable ordinal"
by a "condensation" is still a
"countable ordinal".

\AP A ""condensation formula""~$\varphi(x,y)$ is a formula which is interpreted as a condensation of the domain over all "countable ordinal words",
"ie" for every word $w\in\alphabet\ord$, the relation
defined on $\dom(w)$ by $\iota \condensation_{\varphi} \kappa$
if and only if $w, [x \mapsto \iota,y \mapsto \kappa] \models \varphi(x,y)$
is a condensation.
A "condensation formula"~$\varphi(x,y)$ induces a map:
\begin{align*}
	\hat\varphi\colon \alphabet\ord \to (\alphabet\ordp)\ord
\end{align*}
where for every $u\in \alphabet\ord$,
$\hat\varphi(u)$ is a word whose domain is
$\dom(w)/{\condensation_{\varphi}}$, and such that for every
class $I \in \dom(w)/{\condensation_{\varphi}}$,
the $I$-th letter of $\hat\varphi(u)$ is the word
$(u_\iota)_{\iota \in I}$---hence
$\flatten(\hat\varphi(u)) = u$.

\AP For example, the formula $\finite(x,y)$ is a "condensation formula",
called ""finite condensation"". The function
$\hat{\varphi}_{\finite}\colon
\alphabet\ord \to (\alphabet\ord)\ord$ that it induces
sends the word
$ababab\cdots cdcdcd\cdots abc \in \alphabet\ord$
of length $\omega \timesord 2 \plusord 3$ to the 3-letter word
\[(ababab\cdots)(cdcdcd\cdots)(abc).\]
Observe that for every "word@@ord" $w\in \alphabet\ord$,
every letter of $\hat{\varphi}_{\finite}(w)$ is
a word of length $\omega$, except possibly for the last letter
(if the word has one), which can be finite.

Given two "\FO-definable functions"---one that describes ``local transformations'' and another that described how to glue these local transformations together---the following lemma allows us
to build a new "\FO-definable function". It is one of the key ingredients
in our proof of \Cref{theorem:main}, and is illustrated in
\Cref{fig:composition-fo-condensation}.

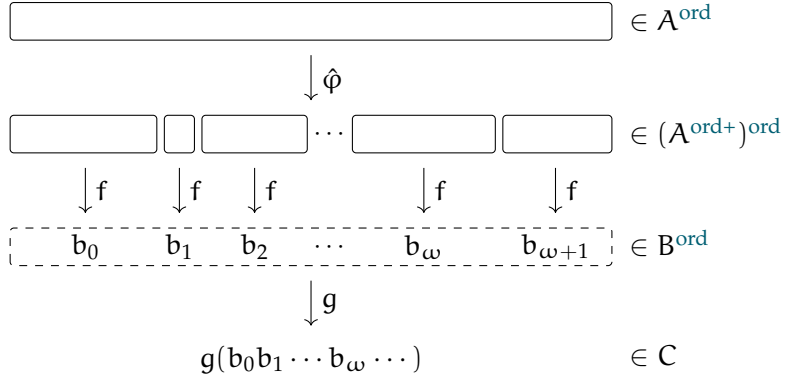
\begin{figure}
	\centering
	\begin{tikzpicture}
    \draw[rounded corners=.5mm] (0,.5) rectangle (8,1);
    \node[right] at (8.1,.8) {$\in A\ord$};

    \draw (4,.3) edge[->] node[midway, right] {$\hat\varphi$} (4,-.3);

    \draw[rounded corners=.5mm] (0,-1) rectangle (1.95,-.5);
    \draw[rounded corners=.5mm] (2.05,-1) rectangle (2.45,-.5);
    \draw[rounded corners=.5mm] (2.55,-1) rectangle (3.95,-.5);
    \node at (4.27,-.75) {\small $\cdots$};
    \draw[rounded corners=.5mm] (4.55,-1) rectangle (6.45,-.5);
    \draw[rounded corners=.5mm] (6.55,-1) rectangle (8,-.5);

    \node[right] at (8.1,-.75) {$\in (A\ordp)\ord$};

    \draw[dashed,rounded corners=.5mm] (0,-2.5) rectangle (8,-2);
    \node at (1,-2.25) {$b_0$};
    \draw (1,-1.2) edge[->] node[midway, right] {$f$} (1,-1.8);
    \node at (2.25,-2.25) {$b_1$};
    \draw (2.25,-1.2) edge[->] node[midway, right] {$f$} (2.25,-1.8);
    \node at (3.25,-2.25) {$b_2$};
    \draw (3.25,-1.2) edge[->] node[midway, right] {$f$} (3.25,-1.8);
    \node at (4.27,-2.25) {\small $\cdots$};
    \node at (5.5,-2.25) {$b_\omega$};
    \draw (5.5,-1.2) edge[->] node[midway, right] {$f$} (5.5,-1.8);
    \node at (7.25,-2.25) {$b_{\omega+1}$};
    \draw (7.25,-1.2) edge[->] node[midway, right] {$f$} (7.25,-1.8);

    \node[right] at (8.1,-2.2) {$\in B\ord$};

    \draw (4,-2.7) edge[->] node[midway, right] {$g$} (4,-3.3);
    \node at (4,-3.75) {$g(b_0b_1\cdots b_\omega \cdots)$};	
    \node[right] at (8.1,-3.7) {$\in C$};
\end{tikzpicture}
	\caption{The map $g \focomp{\varphi} f$ defined in
	\Cref{lemma:composition-fo-condensation} consists in applying $g$ globally
	after applying $f$ locally to the word induced by the condensation 
	$\varphi$.}
	\label{fig:composition-fo-condensation}
\end{figure}

\begin{lemma}
	\labelwithproof{lemma:composition-fo-condensation}
	Let $A,B,C$ be finite sets.
	Let~$\varphi(x,y)$ be a "condensation \FO-formula" over~$A$,
	let $f\colon A\ordp\rightarrow B$ and $g\colon B\ord \rightarrow C$
	be "\FO-definable functions".
	Then, the map\phantomintro\focomp
	\begin{align*}
		g \reintro*\focomp{\varphi} f \colon A\ord &\to \qquad C\\
			u&\mapsto g\left(\prod\limits_{i\in dom(\hat\varphi(u))} f(\hat\varphi(u)_i)\right)
	\end{align*}
	is "\FO-definable@@map".
\end{lemma}

\section{The algorithm}
\label{section:algorithm}
% !TeX root = ../main-separation-ordinals.tex
% !TEX root = ../main-separation-ordinals.tex

In this section we describe the "algorithm" behind \Cref{theorem:main}.
We first introduce the key notion of "saturation" in \Cref{subsection:saturation}, and formalise the "algorithm" in \Cref{subsection:algorithm}.

\subsection{The saturation construction}
\label{subsection:saturation}
Until the end of \Cref{subsection:saturation}, we fix a finite "ordinal monoid"~$\ordmonoid=(M,\Acdot,\Aunit,-\Aomega)$.

The "saturation" construction is at the heart of the "algorithm", both in this paper, and in previous work. 
We introduce the necessary definitions. Note however that in our case, we do not close the definition under subsets as is usually done.
This change, which may look minor, is in fact key for our proof to go through in the case of "countable ordinals", and we find it also simplifies some points in the setting of "finite words".
We first recall an essential operation on $\Pset(M)$ that we denote $-\Pmerge$. Applied to a set~$X\subseteq M$, it computes the union of all the elements that belong to the maximal group in the subsemigroup of $\Pset(M)$ generated by~$X$.
\begin{definition}\AP\phantomintro{\Pmerge}%
	Let~$X\subseteq M$. Define
	\begin{align*}
		X\reintro*\Pmerge=\bigcup\limits_{k\in\mathbb{N}} X\idem[+k] \mathbin{{=}^\star}
		\bigcap_{n\in\Nats} \bigcup_{m \geq n} X^{m}.
	\end{align*}
\end{definition}
\noindent Note that the $\star$ equality holds: Left to right inclusion comes from the fact that
$X\idem[+k] = X^m$ holds for infinitely many values of~$m$, while the other inclusion stems from the fact that~$X^m$ can be written as $X\idem[+k]$ for some~$k$ whenever $m$ is sufficiently large.

Some important properties of this operation are the following.
\begin{lemma}\labelwithproof{lemma:properties-Pmerge}
	The operation~$-\Pmerge$ is monotonic, and for all~$A,B\subseteq M$, and all integers~$k$,
	\begin{align*}
		A\idem[+k]&\subseteq A\Pmerge, &
		(A\Pcdot B)\Pmerge &= A\Pcdot (B\Pcdot A)\Pmerge\Pcdot B\ ,\\
		\text{and}\qquad
		A\Pmerge\Pcdot A\Pmerge&=(A\Pmerge)\Pmerge=A\Pmerge.
	\end{align*}
\end{lemma}

\AP The core of the "algorithm" computes the closure under $-\Pmerge$ and all the operations of the algebra of the images of the letters.
\begin{definition}\label{definition:sat}\label{definition:merge}%
	Let $A\subseteq \powerset(M)$. The set $\intro*\Closgord A\subseteq \powerset(M)$ is defined to be 
	the least set containing~$A$, $\{\Aunit\}$, and closed under $\Pcdot$, $\Pmerge$ and $\Pomega$.\footnote{Recall that we showed that
	in a "power ordinal monoid", the operation $-\Pomega$ is computable.}
\end{definition}
\AP This definition is close in spirit to what is called  "saturation" in previous works, with the difference that we do not take the downward closure, and that we close under the operation $-\Pomega$. Despite this difference, we sometimes call $\Closgord A$ the ""saturation"".

Observe that the "ordinal monoid" $\ordmonoid$ is "aperiodic"
if and only if
\[
	\Closgord{\{\{x\} \mid x \in \ordmonoid\}} =
	\{\{x\} \mid x \in \ordmonoid\}\ .
\]

\subsection{The algorithm}
\label{subsection:algorithm}%

\AP We are now ready to describe the core of the algorithm that is claimed to exist in \Cref{theorem:main}. Let~$K$ and $L$ be two "regular languages of countable ordinal words" over the alphabet~$\alphabet$.
The ""algorithm"" is:
\begin{enumerate}
\itemAP Let~$\ordmonoid$, $\lettermap$, $F_K$, $F_L$ be such that $K$ is "recognised@@OM" by $(\ordmonoid,\lettermap,F_K)$
and $L$ by $(\ordmonoid,\lettermap,F_L)$.
\item Compute $\intro*\Sat:=\Closgord{\{\{\lettermap(a)\}\mid a\in\alphabet\}}$ (inside~$\Pset(\ordmonoid)$).
\item If $F_K\cap X\neq\varnothing$ and $F_L\cap X\neq\varnothing$ for some~$X\in \Sat$, answer `"no"'. Otherwise answer `"yes"'.
\end{enumerate}\smallskip
%\remi{Maybe move following paragraph at the end this section?}
\begin{figure}
	\begin{center}
		\begin{minipage}{.3\linewidth}
			\scalebox{.8}{% !TEX root =  main-separation-ordinals.tex

\begin{tikzpicture}
    \draw[rounded corners=1pt] (-1,0) rectangle (2,1);  
    \draw (0,0) -- (0,1);
    \draw (1,0) -- (1,1);
    \draw[rounded corners=1pt] (0,1.75) rectangle (1,2.75);
    \draw[rounded corners=1pt] (0,3.5) rectangle (1,4.5);
    % \draw ($(.5,-4)+(.075,-.075)$) rectangle ($(1.5,-3)+(-.075,.075)$);
    % \draw ($(2.5,0)+(.075,-.075)$) rectangle ($(3.5,1)+(-.075,.075)$);
    % \draw ($(-2.5,2)+(.075,-.075)$) rectangle ($(-1.5,3)+(-.075,.075)$);

    \node[outer sep=.6cm] (j-omega) at (.5,.5) {};
    \node[outer sep=.5cm, font=\footnotesize] (awa) at (-.5,.5) {$a^\omega aa_\ast$};
    \node[outer sep=.5cm, font=\footnotesize] (aw) at (.5,.5) {$a^\omega a_\ast$};
    \node[outer sep=.5cm, font=\footnotesize] (aw) at (1.5,.5) {$a^\omega_\ast$};
    \node[outer sep=.4cm] (grp) at (.5,2.25) {$\substack{%
    a,\\ \hphantom{{}_\ast}aa_\ast}$};
    \node[outer sep=.4cm] (1) at (.5,4) {\(\hphantom{{}_\ast}1_\ast\)};
  
    \node[above right = 0cm and 1cm of grp, text width=1cm,font=\scriptsize]
    {\textit{the group $\mathbb{Z}/2\mathbb{Z}$}}
    edge[->, dashed, bend left=20] (grp);

    \draw (j-omega) -- (grp);
    \draw (grp) -- (1);

    % \begin{scope}[xshift=7cm]
    %   \draw (-1,0) rectangle (0,1);  
    %   \draw (0,0) rectangle (1,1);
    %   \draw (1,0) rectangle (2,1);
    %   \draw (4,0) rectangle (5,1);
    %   \draw (2,-1.75) rectangle (3,-.75);
    %   \draw (2,1.75) rectangle (3,2.75);
    %   \draw (2,3.5) rectangle (3,4.5);

    %   \node[outer sep=.6cm] (j-omega) at (.5,.5) {};
    %   \node[outer sep=.5cm, font=\small] (awa) at (-.5,.5) {$\hphantom{{}_\ast}a^\omega aa_\ast$};
    %   \node[outer sep=.5cm, font=\small] (aw) at (.5,.5) {$\hphantom{{}_\ast}a^\omega a_\ast$};
    %   \node[outer sep=.5cm, font=\small] (aw) at (1.5,.5) {$\hphantom{{}_\ast}a^\omega_\ast$};
    %   \node[outer sep=.3cm] (grp-mrg) at (4.5,.5) {$\substack{%
    %   \hphantom{{}_\ast}\{a,\hphantom{aa}\\\hphantom{{}_\ast} aa\}_\ast}$};
    %   \node[outer sep=.4cm] (grp) at (2.5,2.25) {$\substack{%
    %   \hphantom{{}_\ast}a,\hphantom{a}\\ \hphantom{{}_\ast}aa_\ast}$};
    %   \node[outer sep=.4cm] (omega-mrg) at (2.5,-1.25) {$\substack{%
    %     \{a^\omega a,\\ a^\omega aa\}_\ast}$};
    %   \node[outer sep=.5cm, font=\small] (1) at (2.5,4) {$\hphantom{{}_\ast}1_\ast$};

    %   \draw (omega-mrg) -- (j-omega) -- (grp) -- (1);
    %   \draw (omega-mrg) -- (grp-mrg) -- (grp);

    %   \node[above right = .75cm and 0cm of grp-mrg, text width=1.5cm,font=\scriptsize]
    %   {\textit{element obtained using the \(-\Pmerge\) operation}}
    %   edge[->, dashed, bend right=20] (grp-mrg);
    % \end{scope}
  \end{tikzpicture}
  }
		\end{minipage}
		% \hspace{.05\linewidth}
		\begin{minipage}{.6\linewidth}
		\small
		\begin{tabular}{cc}
			  \setlength{\arraycolsep}{5pt}
			$\begin{array}{c|cccccc}
			\Acdot & 1 & a & aa & a^\omega & a^\omega a & a^\omega aa\\
			\hline
			1 & 1 & a & aa & a^\omega & a^\omega a & a^\omega aa\\
			a & a & aa & a & a^\omega &a^\omega a & a^\omega aa \\
			aa & aa & a & aa &a^\omega & a^\omega a & a^\omega aa \\
			a^\omega & a^\omega & a^\omega a & a^\omega aa & a^\omega &a^\omega a & a^\omega aa\\
			a^\omega a & a^\omega a & a^\omega aa & a^\omega a&a^\omega&a^\omega a&a^\omega aa\\
			a^\omega aa & a^\omega aa &a^\omega a& a^\omega aa&a^\omega&a^\omega a&a^\omega aa\\
			\hline
			-\Aomega&1&a^\omega&a^\omega&a^\omega&a^\omega&a^\omega
			\end{array}$
		\end{tabular}
		\end{minipage}\vspace{.5cm}
		\begin{center}\(
			\Closgord{\{\{a\}\}}=\{\{1\},\{a\},\{aa\},\{a,aa\},\{a^\omega\},\{a^\omega a\},\{a^\omega aa\},\{a^\omega a,a^\omega aa\}\}
		\)\end{center}
		\end{center}
	\caption{\label{fig:tiptop-monoid}%
	Egg-box diagram of a finite "ordinal monoid" $\ordmonoid$
	recognising~$\exJ$, $\exK$ and $\exL$ (left),
	multiplication table and "\(\omega\)-iteration" of \(\ordmonoid\)
	(right) and "saturation" (bottom).}
\end{figure}
% \begin{figure}
% 	\centering
% 	\begin{tabular}{cc}
% 	$\begin{array}{c|cccccc}
% 		&1&a&aa&a^\omega&a^\omega a&a^\omega aa\\
% 	\hline
% 	1 &1&a&aa&a^\omega&a^\omega a&a^\omega aa\\
% 	a &a&aa&a&a^\omega a&a^\omega aa&a^\omega a\\
% 	aa &aa&a&aa&a^\omega aa&a^\omega a&a^\omega aa\\
% 	a^\omega &a^\omega&a^\omega a& a^\omega aa&a^\omega&a^\omega a&a^\omega aa\\
% 	a^\omega a&a^\omega&a^\omega aa& a^\omega a&a^\omega&a^\omega a&a^\omega aa\\
% 	a^\omega aa&a^\omega&a^\omega a& a^\omega aa&a^\omega&a^\omega a&a^\omega aa\\
% 	\hline
% 	-\Aomega&1&a^\omega&a^\omega&a^\omega&a^\omega&a^\omega
% 	\end{array}$
% 	\\
% 	$	\Closgord{\{a\}}=\{\{1\},\{a\},\{aa\},\{a,aa\},\{a^\omega\},\{a^\omega a\},\{a^\omega aa\},\{a^\omega a,a^\omega aa\}\}
% 	$
% %	&
% %	\scalebox{.8}{\input{tiptop-monoid}}
% 	\end{tabular}
% 	\caption{\label{fig:tiptop-monoid}%
% 	A finite "ordinal monoid" $\ordmonoid$
% 	recognising~$\exJ$, $\exK$ and $\exL$ from \Cref{example:saturation}.}
% \end{figure}
\begin{example}
	\label{example:saturation}
	\AP We illustrate the saturation construction and the algorithm
	on the following three languages over the singleton alphabet $\{a\}$:
	\phantomintro\exJ\phantomintro\exK\phantomintro\exL
	\begin{align*}
		\reintro*\exJ &= \{\text{infinite "words@@ord" whose
		longest finite suffix has even length}\},\\
		\reintro*\exK &= \{\text{infinite "words@@ord" whose
		longest finite suffix has odd length}\}, \\
		\text{and}\qquad\reintro*\exL &=  \{\text{\hphantom{infinite} "words@@ord" that do not have a last letter}\}.
 	\end{align*}
	%Consider the finite "ordinal monoid" $\ordmonoid'$
	%with fives elements: $1$, $a$, $aa$, $a\Aomega$
	%and $a\Aomega a$, satisfying the identity $(a\Aomega)\Aomega = a$.
	%Let $\lettermap\colon\alphabet\to M$ be the map sending
	%$a\in\alphabet$ to $a \in M$. Then
	%$\exJ$ is "recognised@@OM" by
	%$(\ordmonoid_1,\lettermap,\{a\Aomega\})$
	%while $\exK$ is "recognised@@OM" by
	%$(\ordmonoid_1,\lettermap,\{a\Aomega a\})$.
	%It can be shown that $\ordmonoid'$ is the syntactic "ordinal 
	%monoid"\footnote{This means that it is a subsemigroup of a morphic 
	%image of every "ordinal monoid" recognising those languages.}
	%of both $\exJ$ and $\exK$.
	%Since $\ordmonoid'$ is not 
	%"aperiodic"---it contains the cyclic group $\{a,aa\}$---, it follows 
	%that neither $\exJ$ nor $\exK$ can be recognised by a finite
	%"aperiodic" "ordinal monoid", and hence, by "Bedon's theorem",
	%those languages are not "\FO-definable@@lang".
	%On the other hand, $\exL$ is be "defined@@lang" by the \FO-sentence
	%$\neg \exists x.\,\last(x)$
	%where \AP $\intro*\last(x) := \forall y.\, y \leq x$.
	It is classical that~$\exJ$ and~$\exK$ are not "\FO-definable@@lang", while $\exL$ is defined by the formula~$\forall x.\,\exists y.\, y>x$.	We can build a finite "ordinal monoid" $\ordmonoid$ "recognising@@OM" all 
	three languages: it has six
	elements, $1$, $a$, $aa$,
	$a^\omega$, $a^\omega a$ and $a^\omega aa$.
	Its "presentation@@OM" its described \Cref{fig:tiptop-monoid}.
	Naturally, the letter~$a$ is mapped to~$\lettermap(a)=a$.
	Then $\exJ$, $\exK$ and $\exL$ are recognised by
	$F_{\exJ} := \{a^\omega, a^\omega aa\}$,
	by $F_{\exK} := \{a^\omega a\}$
	and by $F_{\exL} := \{1,a^\omega\}$, respectively.

	The languages $\exK$ and $\exL$ are "\FO-separable":
	in fact $\exL$ is an "\FO-separator" of $\exK$ and $\exL$.
	On the other hand, $\exJ$ and $\exK$ are not
	"\FO-separable", as witnessed by the "saturation" algorithm.
	Indeed, the "saturation"
	$\Closgord{\{\{\lettermap(a)\} \mid a \in \alphabet\}}$
	contains all singletons, and furthermore $\{a,aa\} = \{a\}\Pmerge$. As a consequence, it also contains
	$\{a^\omega a, a^\omega aa\} = \{a\}\Pomega \Pcdot \{a,aa\}$.
	This last set intersects both $F_{\exJ}$ and $F_{\exK}$.
\end{example}

The rest of the paper is dedicated to establishing the validity of this approach.
In \Cref{section:no}, we prove \Cref{proposition:no} stating that if the "algorithm" answers `"no"', then the languages cannot be separated, as described in \Cref{theorem:main}.
In \Cref{section:yes}, we prove \Cref{corollary:yes} stating that if the "algorithm" answers `"yes"', then it is possible to construct an "\FO-separator sentence" as described in \Cref{theorem:main}.
In \Cref{section:related}, we shall package the results of \Cref{section:no,section:yes} differently, concluding that we have in fact computed the "pointlike sets", and that we can also decide the more general "covering problem".

\section{When the algorithm says `"no"'}
\label{section:no}
% !TeX root = ../main-separation-ordinals.tex
% !TEX root = ../main-separation-ordinals.tex

In this section, we establish the correctness of the "algorithm", i.e., when the "algorithm" answers `"no"',
we have to prove that the two input languages cannot be separated by an "\FO-definable language", and that we can produce a "witness function".
This is established in
\Cref{proposition:no}.
The proof follows standard arguments.

\AP The ""quantifier depth"", a.k.a. quantifier rank, of an "\FO-formula" is the maximal number of
nested quantifiers in the formula.
\AP Two words~$u,v\in \alphabet\ord$ are said to be ""\FOk-equivalent"",
denoted by $u \intro*\FOkeq v$, if every 
"\FO-sentence" of "quantifier depth" at most~$k$ "accepts"~$u$ if and only if it "accepts"~$v$.

\begin{proposition}%
	\AP\label{prop:compositionality-FO}%
	\label{prop:rosenstein-main-text}%
	Let $k\in\Nats$.
	\begin{itemize}
		\item For $u,u',v,v' \in \alphabet\ord$,
			if $u \FOkeq u'$ and $v \FOkeq v'$ then $uv \FOkeq u'v'$,
		\item for all $\alphabet\ord$-valued sequences $(u_n)_{n\in\Nats}$
		and $(v_n)_{n\in\Nats}$, if $u_n \FOkeq v_n$ for all $n\in\Nats$,
		then $\flatten(u_n \mid n\in\Nats) \FOkeq
		\flatten(v_n \mid n\in\Nats)$, and
		\item for all $n\geqslant 2^k-1$, for all $u\in \alphabet\ord$,
			$u^n \FOkeq u^{n+1}$.
	\end{itemize}
\end{proposition}

This can be proved, for example,
by using Ehrenfeucht-Fraïssé games---see e.g.
\cite[Lemma 6.5 \& Corollary 6.9]{rosenstein1982linear} for a proof of
the first and third items; the proof of the second item is similar\footnote{Moreover, note that the first item can be deduced from the second item
by taking $u_n = v_n = \varepsilon$ for $n \geq 2$.}.
Note that the first two items are also immediate corollaries of the 
Feferman-Vaught theorem \cite[Theorem 1.3]{makowsky2004algorithmic}.
Note that the third property can
be used to prove that every "\FO-definable language"
is recognised by an "aperiodic" finite "ordinal monoid"---this is the 
easy direction of "Bedon's theorem" \cite{bedon2001logic}.

\AP
Throughout the rest of this section, we fix~$K$ and $L$, two "regular languages of countable ordinal words" over an alphabet~$\alphabet$. Recall that the "algorithm" computes the subset $\reintro*\Sat := \Closgord{\{\{\lettermap(a)\}\mid a\in\alphabet\}}$ of $\Pset(\monoid)$, where~$\monoid$ is a finite "ordinal monoid" recognizing both $K$ and $L$.

We begin with a lemma which states that to all sets that belong to $\Sat$ can be effectively associated witnesses of indisinguishability
(we shall see in \Cref{proposition:pointlikes-computable} that what we have proved is that the elements in~$\Sat$ are "pointlike sets").
\begin{lemma}\AP\labelwithproof{lemma:saturation-sets-are-pointlike}
	There exists a computable function which takes as input a number $k \in \Nats$ and an element $X \in \Sat$,
	and produces an $X$-indexed sequence of ordinal words $(u_x)_{x \in X} \in (\alphabet\ord)^X$
	such that, 
	\begin{itemize}
	\item $\product(\ordmap(u_x)) = x$ for all~$x\in X$, and
	\item $u_{x} \FOkeq u_{x'}$ for all~$x,x'\in X$.
	\end{itemize}
\end{lemma}
The proof is by structural induction on the definition of~$\Sat$, making use of the two first items of \Cref{prop:compositionality-FO}
for composing witnesses, and of furthermore the third item for treating the $-\Pmerge$ operation.

From the above lemma, one can easily deduce that when the "algorithm" answers `"no"', there is indeed an obstruction to the fact that~$K$ and~$L$ can be "\FO-separated".
\begin{proposition}\AP\label{proposition:no}
	Assume that the "algorithm" answers `"no"' when run with input languages $K$ and $L$. Then there is a "witness function" which computes, for any "\FO-sentence" $\varphi$, a pair of words $(u,u') \in K \times L$ such that $u \models \varphi$ if and only if $u' \models \varphi$. In particular, $K$ and $L$ cannot be "\FO-separated".
\end{proposition}
\begin{proof}
	Since the algorithm answered `"no"', pick a pair $(x, x') \in F_K \times F_L$ such that $x, x' \in X$ for some $X \in \Sat$. Now, for any \FO-sentence $\varphi$, using the function of \Cref{lemma:saturation-sets-are-pointlike} with $k$ the quantifier depth of $\varphi$, we can compute a sequence $(u_x)_{x \in X}$ of ordinal words. Now define $u := u_x$ and $u' := u_{x'}$. Then $u \FOkeq u'$, so that $u \models \varphi$ if and only if $u' \models \varphi$. Also, $\product(\ordmap(u)) = x \in F_K$ and $\product(\ordmap(u')) = x' \in F_L$, so $u \in K$ and $u' \in L$.
%\sam{To be filled in, but this should be clear from the new formulation of Lemma~\ref{lemma:saturation-sets-are-pointlike}, I sketch the idea because I am out of time: since the algorithm says no, pick a pair $(x,x') \in F_K \times F_L$ which lie in a common $X \in S$, take $k$ the quantifier depth of $\phi$ and use the pair of words $(u,u')$ produced in the lemma, they will be witnesses that $\phi$ does not separate $K$ from $L$.}.
\end{proof}

\begin{example}[Continuing \Cref{example:saturation}]\AP\label{example:no}
	Recall that \(\exJ\) and \(\exK\) are not "\FO-separable". Because
	of the set \(\{a^\omega a, a^\omega aa\} \in
	\Closgord{\{\lettermap(a)\ \mid a \in \alphabet\}}\),
	the "algorithm" outputs `"no"', and can return, to "witness@witness function" the
	"\FO-inseparability" of the two languages the computable map
	\[\varphi \mapsto (a^\omega a^{2^k+1},\, a^\omega a^{2^k+2}) \in \exJ \times \exK,\]
	where $k$ denoted the "quantifier depth" of $\varphi$.
	To prove that $a^\omega a^{2^k+1} \FOkeq a^\omega a^{2^k+2}$, one
	can simply use the first and third items of \Cref{prop:compositionality-FO}.
\end{example}

\section{When the algorithm says `"yes"'}
\label{section:yes}
% !TeX root = ../main-separation-ordinals.tex
% !TEX root =  ../main-separation-ordinals.tex

We now establish the completeness part of the proof of the main theorem, \Cref{theorem:main}.
The goal of this proof is to establish that if the "algorithm" answers `"yes"', it is indeed possible to produce an "\FO-separator" (\Cref{corollary:yes}).

This is the part of the proof that differs most substantially from previous works on separation.
In \Cref{subsection:merge-and-approximants}, we abstract the question with the notion of "ordinal monoids with merge",
and we introduce the notion of "\FO-approximants" which are "\FO-definable@@map"
over-approximations of the product. The key result, \Cref{lemma:completude-core}, states their existence for all finite "ordinal monoid with merge".  \Cref{corollary:yes} follows immediately.
The proof of \Cref{lemma:completude-core} is then established in  \Cref{section:FO-approximant-finite-and-omega} for words of "finite@finite words" or  "$\omega$ length@$\omega$-words". Building on these simpler cases, the general case is the subject \Cref{section:FO-approximant-ordinal}.

\subsection{"Merge operators" and "{\FO}-approximants"}
\label{subsection:merge-and-approximants}

\AP We abstract in this section the ordinal $\Pset(M)$ equipped with the $-\Pmerge$ operator into a new algebraic structure.
A finite ""ordinal monoid with merge"" $\monoid=(M,\Aunit,\leqslant,\Acdot,\Aomega,\Amerge)$ consists of:
\begin{itemize}
\item a "presentation@@OM" of an "ordered ordinal monoid" $(M,\Aunit,\leqslant,\Acdot,\Aomega)$, together with
\item a monotonic ""merge operator""~$-\intro*\Amerge\colon M\to M$ such that for all~$a,b\in M$, and all integers~$k$,
	\begin{align*}
		a\idem[+k]&\leqslant a\Amerge, &
		(a\idem)\Amerge &= a\idem,\\
		a\Amerge\Acdot a\Amerge&=(a\Amerge)\Amerge=a\Amerge,&
		\text{and}\qquad (a\Acdot b)\Amerge &= a\Acdot(b\Acdot a)\Amerge \Acdot b\ .
	\end{align*}
\end{itemize}
The following lemma is an immediate consequence of  \Cref{lemma:properties-Pmerge}.
\begin{lemma}\AP\label{lemma:powerset-is-merge}
	$(\Pset(\ordmonoid),\{\Aunit\},\subseteq,\Pcdot,\Pomega,\Pmerge)$ and $(\Sat,\{\Aunit\},\subseteq,\Pcdot,\Pomega,\Pmerge)$ are "ordinal monoids
	with merge".
\end{lemma}
The idea behind "ordinal monoids with merge" is that not only there is a product operation as for every "ordinal monoid", but also an "\FO-definable@@map"
over-approximation for it. This is the concept of "\FO-approximant" that we introduce now.
\AP Given a an "\FO-definable language"~$L\subseteq M\ord$, an ""\FO-approximant"" of~$\product$ over~$L$ is an "\FO-definable map"~$\rho\colon L\to M$ such that:
\begin{align*}
	\product(u)\leqslant \rho(u),\qquad\text{for all~$u\in L$.}
\end{align*}
The key result concerning "ordinal monoids with merge" is the existence of a total "\FO-approximant":
\begin{lemma}\AP\label{lemma:completude-core}
	There is an "\FO-approximant" $\intro*\PFOproduct$ over
	 $M\ord$ for all "ordinal monoids with merge"~$\monoid$.
\end{lemma}

An example of an "\FO-approximant" can be found in \Cref{example:approximant}. Before establishing  \Cref{lemma:completude-core}, let us explain why it is sufficient for  concluding the proof of \Cref{theorem:main} in the case the "algorithm" answers `"yes"'.
\begin{corollary}\AP\label{corollary:yes}
	If the "algorithm" answers `"yes"', there exists an "\FO-separator".
\end{corollary}
\begin{proof}
	By \Cref{lemma:powerset-is-merge,lemma:completude-core}, there exists an "\FO-approximant" \[\PFOproduct: A\ord \to \Closgord A\] over the "power 
	ordinal monoid" $\Pset(\ordmonoid)$,
	where $A = \{\{\lettermap(a)\} \mid a \in \alphabet\}$.
	%Consider the function $\psi: \alphabet\ord \to \Closgord{A}$
	%defined as $\PFOproduct\circ \tilde{\varphi}\ord$,
	%where $\tilde\varphi: \alphabet \to A$
	%is defined by $\tilde\varphi(a) := \{\varphi(a)\}$ for all
	%$a\in\alphabet$.
	%It is "\FO-definable@@map" as the precomposition of an "\FO-definable@@map"
	%function with a letter-to-letter morphism, and satisfies
	%$\varphi(u) \subseteq \psi(u)$ for all $u\in \alphabet\ord$.
	\AP Now define the language
	\begin{align*}
		S &:= \{u\in\alphabet\ord\mid \PFOproduct(\singordmap(u))\cap F_K\neq\varnothing\} \\
		\text{ where } {\intro*\singordmap}(u) &:= (\{\lettermap(u_i)\})_{i\in \dom(u)} \in A\ord \text{ for all }u\in\alphabet\ord.
	\end{align*}
	Note first that since~$\PFOproduct$ is "\FO-definable@@map", this language is "\FO-definable@@lang". Let us show that it "separates"~$K$ from~$L$.

	For every $u \in K$, $F_K\ni\product(\ordmap(u)) \subseteq \PFOproduct(\singordmap(u))$, and 
	as a consequence $\PFOproduct(\singordmap(u))\cap F_K \neq \varnothing$. We have proved $K \subseteq S$.
	
	Conversely, consider some $u \in L$.
	We have
	\[ F_L\ni\product(\ordmap(u))\in\PFOproduct(\singordmap(u))\in \Closgord A,\]
	and thus $\PFOproduct(\singordmap(u))\cap F_L\neq\varnothing$.
	Since the "algorithm" returns `"yes"', this means that there is no set in~$\Closgord A$ that intersects both~$F_K$ and $F_L$.
	In our case, this means that $\PFOproduct(\singordmap(u))\cap F_K=\varnothing$, proving that $u\not\in S$.
	We have proved $L\cap S=\varnothing$.
	
	Overall, $S$ is an "\FO-separator" for~$K$ and~$L$.
\end{proof}

\begin{remark}
%\remi{This paragraph is new!}\thomas{I made some changes also.}
Notice how the ``difficult'' implication of "Bedon's theorem" (\Cref{prop:bedon-thm}) can be easily deduced
from \Cref{lemma:completude-core}\footnote{Similarly, for finite words, Schützenberger-McNaughton-Papert's theorem
is a consequence of Henckell's algorithm for aperiodic pointlikes---see e.g. \cite[Corollary 4.8]{zeitoun2016separating}}:
recall that this implication consists in showing that a
"regular language@@cow" $L \subseteq \alphabet\ord$, "recognised@@OM" by some triplet 
$(\ordmonoid,\lettermap,F)$ with $\ordmonoid$ is "aperiodic" is definable in "first-order logic".
Indeed, by "aperiodicity" of $\ordmonoid$, the operation~$\Pmerge$ applied to a singleton~$\{a\}$
yields the singleton~$\{a\idem\}$. Hence, the set
$\Closgord{\{\{\lettermap(a)\}\
\mid a \in \alphabet\}} = \{\{\product\circ\ordmap(u)\} \mid u\in 
\alphabet\ord\}$ consists only of singletons, and as a consequence, all "\FO-approximants" $\rho$ (and in particular the one constructed in \Cref{lemma:completude-core})
maps a word~$u$ to~$\product(u)$. 
Hence, $\product$ is an "\FO-definable map", and thus~$L$ is an "\FO-definable language". 
%Another way to prove the result consists in remarking that if~$\ordmonoid$ is an "aperiodic" "ordinal monoid",
%then it can be equipped with the equality as order, and $-\idem$ as $-\Amerge$-operation, yielding an "ordinal monoid with merge".
\end{remark}

The rest of this section is devoted to establishing
\Cref{lemma:completude-core}. The construction is based on subresults showing 
the existence of "\FO-approximants" over subsets of~$M\ord$; first for "finite@finite words" 
and "$\omega$-words" in \Cref{section:FO-approximant-finite-and-omega}, and 
finally for "words of any countable ordinal length@words@ord" in \Cref
{section:FO-approximant-ordinal}. But beforehand, we shall introduce some more 
definitions and elementary results.

%\subsection{General properties}
%\label{section:FO-approximant-elementary}

\AP
In what follows we use the notation $\Closgord -$ from \Cref{definition:merge}, interpreted in a generic "ordinal monoid with merge", as well as some variants. Let~$A\subseteq M$.
We define $\intro*\Closp A$ as the closure of~$A$ under~$\Acdot$, $\intro*\Closgp A$ as the closure of~$A$ under $\Acdot$ and~$-\Amerge$, and~$\intro*\Closgs A$ as~$\Closgp A\cup \{\Aunit\}$. 
We define~$\intro*\Closgordp A$ as the closure of~$A$ under~$\Acdot$, $\Amerge$ and $\Aomega$.
Note that thanks to the identities of "ordinal monoids with merge", we have $\Closgord A=\Closgordp A\cup\{\Aunit\}$.
Moreover, we have the following identities\footnote{Notice the similarity with the (trivial) identities $A^+ = AA^* = A^*A$ and
$A\ordp = A A\ord$.}:

\begin{proposition}\AP\label{prop:saturation-take-letter-out}
	Let $\ordmonoid$ be an "ordinal monoid with merge".
	For every $A \subseteq \ordmonoid$, 
	\begin{align*}
		\Closgp A &= A \Closgs A = \Closgs A A&
		\text{and}\qquad\Closgordp A &= A \Closgord A\ .
	\end{align*}
\end{proposition}
\begin{proof}
	Note, by definition, that~$\Closgs A = \Closgp A \cup \{\unit\}$, 
	so $A \Closgs A = A \Closgp A \cup A \subseteq \Closgp A$.
	The converse inclusion $\Closgp A \subseteq A \Closgs A$ is obtained by induction. Let~$b\in\Closgp A$. If~$b\in A$,
	then $b\in A\Closgs A$ since~$\unit\in \Closgs A$.
	If~$c=cd$ with~$c,d\in \Closgp A$, then, by induction,
	$c=ac'$ for some~$a\in A$ and~$c'\in\Closgs A$, thus~$b=a(c'd)\in A\Closgs A$
	since~$a\in A$ and~$c'd\in\Closgs A$. Finally, if~$b=c\Amerge$,
	then, again by induction, $c=ac'$ for some~$a\in A$ and~$c'\in\Closgs A$, and
	thus~$b=c\Amerge=cc\Amerge = a(c'c\Amerge)\in A\Closgs A$.
	
	The equality 	$\Closgp A = \Closgs A A$ is symmetric.
	
	The identity $\Closgordp A = A \Closgord A$ is similar. The new case in the induction is if some~$b\in\Closgordp A$
	is of the form~$c^\omega$, then, by induction hypothesis, $c=ac'$ for some~$a\in A$ and~$c'\in\Closgord A$, and
	thus~$b=c^\omega=cc^\omega = a(c'c^\omega)\in A\Closgord A$.
\end{proof}

\begin{proposition}\AP\label{prop:fo-union}\label{prop:fo-concat}\labelwithproof{prop:fo-concat-union}
	If there are "\FO-approximants" over~$K$ and~$L$ respectively, then there exist effectively "\FO-approximants" over~$K\cup L$ and~$KL$.
\end{proposition}

\subsection{Construction of "{\FO}-approximants" for words of "finite@finite word" and "$\omega$-length@$\omega$-word"}
\label{section:FO-approximant-finite-and-omega}

First, we show how to construct "\FO-approximants" for "finite words".
It serves at the same time as a building block for more complex cases, as a way to show the proof mechanisms in simpler cases, as well as to comment on differences with previous works.
\begin{lemma}\AP\label{lemma:finite-trichotomy}
	Let~$A\subseteq M$, then either
	\begin{itemize}
	\item $a\Acdot \Closgp{A}\subsetneq \Closgp A$, for some~$a\in A$,
	\item $\Closgp{A}\Acdot a\subsetneq \Closgp A$, for some~$a\in A$, or
	\item $\Closgp{A}$ has a maximum.
	\end{itemize}
\end{lemma}
\begin{proof}
	Assume the two first items do not hold.
	Because of the non-first-one, the map $x\mapsto a\Acdot x$ is surjective on~$\Closgp{A}$, for all~$a\in A$.
	Since~$\Closgp{A}$ is finite, this means that it is bijective on $\Closgp{A}$. Hence it is also bijective on~$\Closp{A}$.
	The negation of the second item has a symmetric consequence. Together we get that $\Closp{A}$ is a group.
	Let~$I$ be its neutral element. Note first that for all~$x\in\Closp{A}$, $I=x^k$ for some~$k$, and hence, $I\leqslant x\Amerge$.
	Set now~$a_1,\dots,a_n$ to be the elements in~$A$, and define:
	\begin{align*}
		M&= (a_1\Amerge\Acdot a_2\Amerge\cdots a_n\Amerge)\Amerge\ .
	\end{align*}
	By the above remark~$a_i= I^{i-1}\Acdot a_i\Acdot I^{n-i}\leqslant a_1\Amerge\Acdot a_2\Amerge\cdots a_n\Amerge\leqslant M$ for all~$i$.
	Since furthermore for all $x,y\leqslant M$,  $x\Acdot y\leqslant M$ and $x\Amerge\leqslant M$, it follows that~$z\leqslant M$ for all~$z\in\Closgp{A}$.
\end{proof}
A similar lemma is used in \cite{zeitoun2016separating}, but concludes with the existence of a pseudo-group as the third item.
\begin{lemma}\AP\label{lemma:fo-approximant-astar}
	\AP For all~$a\in M$ there exists an "\FO-approximant"
	from~$a^+$ to~$\Closgp{\{a\}}$.
\end{lemma}
\begin{proof}[Construction] Let~$k$ be such that~$a\idem=a^k$.
	Define
	\begin{align*}
		\rho(\overbrace{a\cdots a}^{\text{length $n$}})&=
			\begin{cases}
			a^n&\text{if $n<k$,}\\
			a\Amerge&\text{otherwise.}
			\end{cases}
			\qedhere
	\end{align*}
\end{proof}
We can now use this for proving the finite word case.
\begin{lemma}\AP\label{fo-approximant-finite}
	\label{lemma:fo-approximant-finite}
	\AP For all~$A\subseteq M$ there exists an "\FO-approximant" from~$A^+$ to~$\Closgp{A}$.
\end{lemma}
\begin{proof}
	We use a double induction on~$|\Closgp A|$ and~$|A|$.
	The induction is guided by \Cref{lemma:finite-trichotomy}.
	The base case is~$A=\varnothing$, and the nowhere defined "\FO-approximant" proves it.
	
	\emph{First case:}~$a\Acdot \Closgp{A}\subsetneq \Closgp A$ for some~$a\in A$.
	This part of the proof is similar to \cite[Lemma 6.7]{zeitoun2016separating}.
	Let~$B ::= A\smallsetminus \{a\}$.

	We first construct an "\FO-approximant" from~$a^+B^+$ to~$a\Acdot \Closgp A$.
	Indeed, we know by  \Cref{lemma:fo-approximant-astar} that there is an "\FO-approximant" from~$a^+$ to  $\Closgp{\{a\}} \subseteq a\Acdot \Closgs A$.
	We also know by induction\footnote{Indeed, $|B| < |A|$.} that there is an "\FO-approximant" from~$B^+$ to~$\Closgp B\subseteq\Closgp A$.
	Thus by \Cref{prop:fo-concat}, there exists effectively an "\FO-approximant"~$\tau$ from~$a^+B^+$ to $a\Acdot \Closgs A\Acdot \Closgp A\subseteq a\Acdot \Closgp{A}$.

	We now provide an "\FO-approximant" for $(a^+B^+)^+$
	(which is "\FO-definable@@lang"), and for this,
	define the "condensation \FO-formula" $\varphi(x,y)$ that
	expresses that ``two positions~$x$ and~$y$ are equivalent
	if the subword on the interval~$[x,y]$ belongs to $a^*B^*$'' (this can be expressed in first-order logic).
	Over a word~$u\in (a^+B^+)^+$, each of the "condensation" classes belong to~$a^+B^+$ and its image under~$\tau$
	belongs to~$a\Acdot \Closgp{A}$. Furthermore, still by induction hypothesis\footnote{This time, we can use the induction
	hypothesis because $|\Closgp{(a\Acdot \Closgp{A})^+}| <
	|\Closgp{A}|$. Indeed, by \Cref{prop:saturation-take-letter-out},
	$\Closgp{(a\Acdot \Closgp{A})^+} \subseteq
	(a\Acdot \Closgp{A})^+ \Closgs{(a\Acdot \Closgp{A})^+}
	\subseteq a\Acdot \Closgp{A} \subsetneq \Closgp{A}$.}, there is an "\FO-approximant"
	from~$(a\Acdot \Closgp{A})^+$ to~$\Closgp{A}$. By \Cref{lemma:composition-fo-condensation},
	we thus obtain an "\FO-definable map" from~$ (a^+B^+)^+$ to~$\Closgp{A}$. It is an "\FO-approximant" by construction.

	Using the above case and \Cref{prop:fo-concat-union}, it can be easily extended to an "\FO-approximant" from~$A^+= AB^*(a^+B^+)^*a^*$ to~$\Closgp{A}$.

	\emph{Second case:}~$\Closgp{A}\Acdot a \subsetneq \Closgp A$.
	This case is symmetric to the first case.

	\emph{Third case:}~$\Closgp A$ has a maximum~$M$.
	Then the constant map that sends every "word"~$u\in A^*$ to~$M$
	is an "\FO-approximant" over~$A^*$.
\end{proof}

\AP Following similar ideas, we can treat the case of "$\omega$-words".
We define here $\intro*\Closgo{A}$ as the elements of the form~$\{a\Acdot b\Aomega\mid a,b\in\Closgp A\}$---or, equivalently,
$\Closgo{A} = (\Closgp A)\Pomega$.

\begin{lemma}\AP\labelwithproof{lemma:fo-approximant-omega}
	Let $M$ be an "ordinal monoid with merge".
	For all $A\subseteq M$, there exists an "\FO-approximant" from~$A^\omega$ to $\Closgo{A}$. 
\end{lemma}

\subsection{Construction of "\FO-approximants" for "countable ordinal words"}
\label{section:FO-approximant-ordinal}

As for the finite case, the proof revolves around a carefully designed case distinction.
This one is more complex to establish, and makes use of "Green's relations" and a precise understanding of the properties 
of "ordinal monoids with merge".
\begin{lemma}[Trichotomy principle]\AP\labelwithproof{lemma:trichotomy-ordinal-words}
	Let~$M$ be a finite "ordinal monoid with merge" 
	and~$A\subseteq M$, then either
	\begin{itemize}
	\item $a\Acdot \Closgordp{A}\subsetneq \Closgordp A$, for some~$a\in A$,
	\item $\Closgordp{\Closgo{A}}\subsetneq\Closgordp A$, or
	\item $x\Acdot y=y$ and $x\Aomega = y\Aomega$, for all $x,y\in\Closgordp A$.
	\end{itemize}
\end{lemma}
The above lemma is key in the proof of the  existence of an "\FO-approximant".
\begin{lemma}\AP\labelwithproof{lemma:fo-approximant-aord}
	For all $a \in \ordmonoid$, there exists an "\FO-approximant"
	over $a\ord$.
\end{lemma}
\noindent The proof follows a similar structure as the one for \Cref{fo-approximant-finite} for the finite case.
This time, \Cref{lemma:trichotomy-ordinal-words} is the key argument that makes the induction progress,
playing the same role as  \Cref{lemma:finite-trichotomy} in the finite case. Note, however, that the second items in \Cref{lemma:finite-trichotomy,lemma:trichotomy-ordinal-words} are very different in structure. And indeed, this entails a different argument for constructing the "\FO-approximant". It is based on performing in one step the condensation of all the maximal factors of order-type $\omega$.

\begin{example}[Continuing \Cref{example:no}]\AP
	\label{example:approximant}
	An "\FO-approximant"
	$\rho\colon a\ord \to \Closgord{\{\lettermap(a)\}}$ 
	of $\product$ over $a\ord$
	in the "ordinal monoid" defined in \Cref{example:saturation}
	can be defined for all~$u\in\{a\}\ord$ as:
	\[
		\rho(u) := \begin{cases}
			\{1\} & \text{if $\dom(u)$ is empty,} \\
			\{a,aa\} & \text{if $\dom(u)$ is finite and non-empty,} \\
			\{a^\omega\} & \text{if $\dom(u)$ is a non-zero
				limit ordinal,} \\
			\{a^\omega a, a^\omega aa\} & \text{if $\dom(u)$
				is an infinite successor ordinal}.
		\end{cases}
	\]
\end{example}

\begin{lemma}\AP\label{lemma:fo-appsoximant-ord}
	For all $A \subseteq \ordmonoid$, there exists
	an "\FO-approximant" from $A\ordp$ to $\Closgordp A$.
\end{lemma}
\begin{proof}
	We prove the result by induction on $|\Closgordp A|$ and $|A\ordp|$.
	The base case $A = \varnothing$ is trivial.
	If $A$ is non-empty, following \Cref{lemma:trichotomy-ordinal-words},
	there are three cases to treat.

	\emph{First case: There exists $a\in A$ such that $a\Acdot \Closgordp{A}\subsetneq \Closgordp A$.}
	This case is as in the proof for "finite words", \Cref{lemma:fo-approximant-finite}, using \Cref{lemma:fo-approximant-aord} in place of \Cref{lemma:fo-approximant-astar}.
	The key reason why the proof remains valid is because the
	hypothesis $a\Acdot \Closgordp{A}\subsetneq \Closgordp A$ implies
	$|\Closgordp{(a\Acdot \Closgordp{A})\ordp}| < |\Closgordp A|$ by
	\Cref{prop:saturation-take-letter-out}\footnote{More precisely, we
	are using the property $\Closgordp{B} = B\Closgord{B}$
	of \Cref{prop:saturation-take-letter-out}. By thinking of elements of
	$\Closgordp{B}$ as ``"countable ordinal words" with merge'', this
	property is simply saying that every ``"countable ordinal word" with merge''
	has a first letter. However, "countable ordinal words" need not have
	a last letter: this is what makes an hypothesis
	of the form $ \Closgordp{A}\Acdot a \subsetneq \Closgordp A$
	unusable---and this is the motivation behind the trichotomy principle
	\Cref{lemma:trichotomy-ordinal-words}.}.

	\emph{Second case}\footnote{Note here that it is different from the second case in the proof of \Cref{lemma:fo-approximant-finite}.}:
	$\Closgordp{\Closgo{A}}\subsetneq\Closgordp A$.
	By \Cref{lemma:fo-approximant-omega}, there is an
	"\FO-approximant" from $A\Womega$ to $\Closgo A$.
	By induction hypothesis\footnote{Indeed, $\Closgordp{\Closgo{A}}\subsetneq\Closgordp A$.}, we have an "\FO-approximant"
	from $(\Closgo A)\ordp$ to $\Closgordp {\Closgo A}
	\subseteq \Closgordp A$.
	Since the formula $\finite(x,y)$ is a "condensation \FO-formula",
	we obtain by \Cref{lemma:composition-fo-condensation}
	an "\FO-approximant" from $(A\Womega)\ordp \to \Closgordp A$.
	Using \Cref{prop:fo-concat-union} and \Cref{lemma:fo-approximant-finite},
	we easily extend it to an "\FO-approximant" from $A\ordp =
	A (A\Womega)\ord A^*$ to $\Closgordp A$.

	\emph{Third case: $x\Acdot y=y$ and $x\Aomega = y\Aomega$, for all $x,y\in\Closgordp A$.}
	Then the product over $A$ sends
	a "countable ordinal word" $u \in A\ordp$ to its last letter
	if the word has a last letter, and to the unique omega power
	of $\Closgordp A$ if the word
	has no last letter. Since the languages of the form
	$A\ordp a$ where $a\in A$ and $\{u \in A\ordp \mid \dom(u) \text{ is a "limit ordinal"}\}$ all are "\FO-definable@@lang",
	it follows that the product over $A$ is "\FO-definable@@map".
\end{proof}

\section{Related problems}
\label{section:related}
% !TeX root = ../main-separation-ordinals.tex
% !TEX root = ../main-separation-ordinals.tex

In this section, we solve two related problems:
the decidability of the "covering problem"
(\Cref{proposition:covering-decidable}), and the computability of "pointlike 
sets"  (\Cref{proposition:pointlikes-computable}). Both are direct applications of the key lemmas presented above.

\AP The ""\FO-covering problem""
asks, given "regular languages@@cow", in our case of "countable ordinal words", $L,K_1,\dots,K_n$, to determine if there exist "\FO-definable languages" $C_1,\dots,C_n$ such that~$L\subseteq \cup_{i}C_i$ and~$C_i\cap K_i=\varnothing$ for all~$i$---see \cite{pz2018covering} for more details. In general, 
 "separation problems" trivially reduce to "covering problems", since $L$
and $K$ are "separable" if and only if there is
a solution to the "covering problem" for the instance $(L,K)$. In the other direction, there is no known example of a variety with decidable separation problem but undecidable covering problem. We show that a further consequence of the above results is that the "\FO-pointlike sets" in a finite "ordinal monoid" (see \Cref{definition:pointlikes}) are computable, from which we deduce:%From the work done above, we in fact obtain:
\begin{proposition}\AP\labelwithproof{proposition:covering-decidable}
	The "\FO-covering problem" for "countable ordinal words" is decidable.
\end{proposition}

Let us now introduce, and explain, the relation with "pointlike sets".
\AP The~""\FOk-closure"" of a word~$u$ is the set $\intro*\closureFOk
{u}$ which contains all words that are "\FOk-equivalent" to~$u$.

\begin{definition}\AP\label{definition:pointlikes}
	\AP Given a finite "ordinal monoid" $\ordmonoid$
	the ""\FO-pointlike sets"" of a map $\lettermap\colon \alphabet \to M$
	are defined by
	\[
		\intro*\PL(\lettermap) ::=
		\bigcap_{k\in\Nats} \down \left\{
		\product(\ordmap(\closureFOk{u})) \mid u \in \alphabet\ord\right\},
	\]
	where $\intro*\down X$ denotes the downward closure of $X$.
\end{definition}

The definition of "pointlike sets" is in fact more general\footnote{In the following discussion, we focus on finite words,
but the notion of variety---of algebras, or of languages---can be extended to
"countable ordinal words" \cite{bedon1998eilenberg} and many other settings \cite[\S 4]{bojanczyk2015recognisable}.}:
given a variety of finite "semigroups" $\mathbb{V}$
one can define a notion of 
"pointlike sets" with respect to this variety. Almeida
observed that the "separation problem"
for the variety $\mathbb{V}$---given two regular languages, can 
they be separated by a $\mathbb{V}$-recognisable language?---is 
decidable if and only if the $\mathbb{V}$-"pointlikes"
of size 2 of every morphism are computable
\cite[Prop. 3.4]{almeida1999some}.
The "covering problem" also has an algebraic counterpart:
it is decidable for the variety $\mathbb{V}$ if and only if,
for every morphism, the collection of all
$\mathbb{V}$-"pointlike sets" of this morphism is computable
\cite[Prop. 3.6]{almeida1999some}\footnote{Beware: there is
a typo in the statement of the first item of the proposition.}.
Hence, the fact that "\FO-covering" and "\FO-separation" are
decidable for finite words is simply a corollary of
Henckell's theorem on "aperiodic pointlikes"
\cite[Fact 3.7 \& Fact 5.31]{henckell1988}, stating that  
they are computable.
Place \& Zeitoun's simpler proof of the decidability of
"\FO-covering" for "finite words" and for "$\omega$-words"
\cite{zeitoun2016separating} relies on the same principle.\footnote{There is a difference in terminology: they refer to the $\PL
(\varphi)$ as ``optimal imprint with respect to \FO{} on
$\varphi$''.} Unsurprisingly, our result can be interpreted in
the same way: we are implicitly showing the following property,
from which one can immediately deduce the computability of $\PL(\lettermap)$.

\begin{proposition}\AP\labelwithproof{proposition:pointlikes-computable}
	\AP Given a finite "ordinal monoid" $\ordmonoid$
	and $\lettermap\colon \alphabet \to M$,
	\begin{align*}
		\PL(\lettermap) = \down\Closgord{\{\{\lettermap(a)\} \mid a\in \alphabet\}}.
	\end{align*}
\end{proposition}

%"Pointlike sets" are defined as subsets of the "ordinal monoid" which are ``together impossible to \FO-separate''. The definition in the context of "ordinal monoids" is the same as for normal monoids---see
%\Cref{definition:pointlikes}.
%Again, 
%See \Cref{appendix-pointlike-sets}.

\section{Conclusion}
\label{section:conclusion}
% !TeX root = ../main-separation-ordinals.tex

In this paper, we have studied the problem of "\FO-separation" over "words of countable ordinal length". Our proof is based on the work of Place and Zeitoun over "words of length~$\omega$" \cite{zeitoun2016separating}. 
  We build an "\FO-approximant" using
essentially the same technique as Place and Zeitoun.
However a key difference is that for "finite words" and "\(\omega\)-words",
the proof relies on a case distinction (\Cref{lemma:finite-trichotomy}) which is
conceptually similar to the characterisation of groups as semigroups
whose translations are bijective. This was no longer sufficient for "countable ordinal words" because
of "\(\omega\)-iterations". In this situation, our new case distinction (\Cref{lemma:trichotomy-ordinal-words}) captures the subtle
interaction of "\(\omega\)-iteration" with groups in finite "ordinal
monoids".
In particular, a difference with previously known algorithms is that we do not close the "saturation" under subset. This a priori innocuous difference has significant consequences on the proof of completeness, yielding some simplifications in the "finite@finite word" and "$\omega$-case@$\omega$-word", and necessary for the proof to be extendable to all "ordinals".
% Another difference of a more pedagogical nature is the introduction of "word expressions" for denoting the "witnesses of inseparability".

Of course, the next step is to go to longer words, in particular "scattered@@word" "countable words", or even better to all "countable words". Here, there are conceptual difficulties, and let us stress also that, starting from "scattered@@word" "countable words", first-order logic and first-order logic with access to Dedekind cuts begin to have a different expressiveness. Thus several notions of separation have to be studied.

\bibliographystyle{alphaurl}
\bibliography{biblio}

\clearpage

\appendix

\section{Proof of preliminary section}
\label{section:proof-prelim}
% !TeX root = ../main-separation-ordinals.tex

\subsection{Proof of \Cref{lemma:composition-fo-condensation}}

\label{proof-lemma:composition-fo-condensation}
\begin{proof}
	Since $g$ is "\FO-definable@@map", for all $c\in C$,
	there exists an "\FO-formula" $\psi_c$ over the alphabet $B$
	that defines $g^{-1}[c] \subseteq Y\ord$.
	Moreover, since $\varphi(x,y)$ is an "\FO-formula",
	and since $f: A\ord \to B$ is "\FO-definable@@map",
	we can build for every $b\in B$ a formula
	$\chi_b(x,y)$ over $X$ that is satisfied if and only if
	the subword indexed by $[x,y\mathclose{[}$ is sent,
	via $f$, to $b\in B$.

	We want to transform the formula $\psi_c$ into a formula
	$\tilde{\psi}_c$ over the alphabet $A$ so that $\tilde{\psi}_c$
	recognises $(g \focomp{\varphi} f)^{-1}[c]$.
	We perform this transformation by structural induction on the
	formula, as follows:
	\begin{itemize}
	\item We replace every existential (resp. universal) quantification 
		over a variable $x$---representing a position $x$
		in a $B$-"word@@ord"---by a pair of existential (resp. 
		universal) quantifiers over
		two new variables $x_\ell$ and $x_r$---representing a convex
		subset $[x_\ell, x_r\mathclose{[}$ of the "domain" of
		an $A$-"word@@ord"---followed by a constraint asking
		the pair of variables $(x_\ell,x_r)$ 
		to define an equivalence class of the "condensation" defined
		by $\varphi$.
	\item We replace every proposition of the form $b(x)$
		by $\chi_b(x_\ell,x_r)$.
	\end{itemize}
	Hence, we have built, for each $c\in C$, a "first-order formula"
	$\tilde{\psi}_c$ that "accepts" $(g \focomp{\varphi} f)^{-1}[c]$:
	it follows that the function  $g \focomp{\varphi} f$ is 
	"\FO-definable@@map".
\end{proof}

\section{Proofs with regard to the algorithm}
\label{section:proof-algo}
% !TeX root = ../main-sepration-ordinals-fossacs.tex

\subsection{Proof of \Cref{lemma:properties-Pmerge}}
\label{proof-lemma:properties-Pmerge}
\begin{proof}
	Let~$k$ be such that~$A^k=A\idem$ for all~$A\subseteq\monoid$.
	In particular, the sequence~$A^n$ is periodic of period~$k$ after position~$k$.
	It follows that~$A\Pmerge=\cup_{n\geq\ell}A^n$ for all~$\ell\geq k$, and in particular, $A^n\Pcdot A\Pmerge=A\Pmerge$,
	and $A\Pmerge=A\Pmerge\Pcdot A^n$ for all naturals~$n$.
	Hence $A\Pmerge=\cup_{n\geq k}(A\Pmerge\Pcdot A^n)=A\Pmerge\Pcdot A\Pmerge$. One also easily gets
	\begin{align*}
		(A\Pcdot B)\Pmerge\Pcdot A
			&=(\cup_{n\geq\ell}(A\Pcdot B)^n)\Pcdot A
			=A\Pcdot (\cup_{n\geq\ell}(B\Pcdot A)^n)
			=A\Pcdot (B\Pcdot A)\Pmerge\ .
	\end{align*}
	Combining with above facts, we get
	$(A\Pcdot B)\Pmerge=A\Pcdot B\Pcdot (A\Pcdot B)\Pmerge=A\Pcdot(B\Pcdot A)\Pmerge\Pcdot B$.
\end{proof}

 \section{Proofs of when the algorithm says `"no"'}
 \label{section:proof-no}
 % !TeX root = ../main-separation-ordinals.tex
% !TEX root = ../main-separation-ordinals.tex

\subsection{Proof of \Cref{lemma:saturation-sets-are-pointlike}}
\label{proof-lemma:saturation-sets-are-pointlike}

\begin{proof}
	Fix $k \in \Nats$. The computable function is defined by structural induction on $X$.	

	\smallskip\noindent\emph{Base case.}
	If $X = \{\lettermap(a)\}$ then we may take $u_{\lettermap(a)} := a$.
		
	\smallskip\noindent\emph{Binary product case.}
	Assume that both $X$ and
	$Y \in \Sat$ satisfy the induction hypothesis: let $(u_x)_{x\in X}$ denote the sequence for $X$ and $(v_y)_{y \in Y}$ the sequence for $Y$. For any $z \in X \cdot Y$, choose $x \in X$ and $y \in Y$ such that $z = x \cdot y$, and define $w_z := u_x v_y$. Clearly, $\product(\ordmap(w_z)) = z$. Also for any $z = xy, z' = x'y' \in X \cdot Y$, we have $u_x \FOkeq u_x'$ and $v_y \FOkeq v_y'$, so it follows by the first item of \Cref{prop:rosenstein-main-text} that $w_z = u_x v_y \FOkeq u_x' v_y' = w_z'$.

	\begin{figure}
		\centering
		\begin{tikzpicture}
    \draw[rounded corners=.5mm] (0,0) rectangle (2,.5);
    \foreach \x in {0,...,7} {
        \draw[rounded corners=.5mm] ($(2,0)+(\x,0)$) rectangle
            ($(3,.5)+(\x,0)$);
    };

    \foreach \x in {1,...,4} {
        \node[font=\footnotesize] at ($\x*(.5,0)+(-.25,.25)$) {$a_{\x}$};
    };
    \foreach \x in {0,...,7} {
        \node[font=\footnotesize] at ($\x*(1,0)+(2.25,.25)$) {$b_{1}$};
        \node[font=\footnotesize] at ($\x*(1,0)+(2.75,.25)$) {$b_{2}$};
    };

    \draw[rounded corners=.5mm] (0,1) rectangle (3.5,1.5);
    \foreach \x in {0,...,3} {
        \draw[rounded corners=.5mm] ($(3.5,1)+1.5*(\x,0)$) rectangle
            ($(5,1.5)+1.5*(\x,0)$);
    };

    \foreach \x in {1,...,7} {
        \node[font=\footnotesize] at ($\x*(.5,0)+(-.25,1.25)$) {$c_{\x}$};
    };
    \foreach \x in {0,...,3} {
        \node[font=\footnotesize] at ($\x*(1.5,0)+(3.75,1.25)$) {$d_{1}$};
        \node[font=\footnotesize] at ($\x*(1.5,0)+(4.25,1.25)$) {$d_{2}$};
        \node[font=\footnotesize] at ($\x*(1.5,0)+(4.75,1.25)$) {$d_{3}$};
    };

    \node[right] at (10,.25) {$\cdots$};
    \node[right] at (9.5,1.25) {$\cdots$};

    \draw (4.75,-.25) edge[->] (4.75,-.75);

    \begin{scope}[yshift=-2.5cm]
        \draw[rounded corners=.5mm] (0,0) rectangle (3.5,0.5);
        \draw[rounded corners=.5mm] (0,1) rectangle (3.5,1.5);
        \foreach \x in {0,1} {
            \draw[rounded corners=.5mm] ($(3.5,1)+3*(\x,0)$) rectangle
                ($(6.5,1.5)+3*(\x,0)$);
        };
        \foreach \x in {0,1} {
            \draw[rounded corners=.5mm] ($(3.5,0)+3*(\x,0)$) rectangle
                ($(6.5,.5)+3*(\x,0)$);
        };
        \draw[rounded corners=.5mm] (10,0) -- (9.5,0) -- (9.5,.5) -- (10,.5);

        \foreach \x in {1,...,4} {
            \node[font=\footnotesize] at ($\x*(.5,0)+(-.25,.25)$) {$a_{\x}$};
        };
        \foreach \x in {0,...,7} {
            \node[font=\footnotesize] at ($\x*(1,0)+(2.25,.25)$) {$b_{1}$};
            \node[font=\footnotesize] at ($\x*(1,0)+(2.75,.25)$) {$b_{2}$};
        };
        \foreach \x in {1,...,7} {
            \node[font=\footnotesize] at ($\x*(.5,0)+(-.25,1.25)$) {$c_{\x}$};
        };
        \foreach \x in {0,...,3} {
            \node[font=\footnotesize] at ($\x*(1.5,0)+(3.75,1.25)$) {$d_{1}$};
            \node[font=\footnotesize] at ($\x*(1.5,0)+(4.25,1.25)$) {$d_{2}$};
            \node[font=\footnotesize] at ($\x*(1.5,0)+(4.75,1.25)$) {$d_{3}$};
        };

        \node[right] at (10,.25) {$\cdots$};
        \node[right] at (9.5,1.25) {$\cdots$};
    \end{scope}

    \begin{scope}[yshift=-2.7cm]
        \draw[decorate, decoration={brace, mirror}] (0,0) --
            node[midway, below, font=\footnotesize] {prefix of length $\ell$} (3.5,0);
        \draw[decorate, decoration={brace, mirror}] (6.5,0) --
            node[midway, below, font=\footnotesize] {period of length $k$} (9.5,0);
    \end{scope}
\end{tikzpicture}
		\caption{\label{fig:synchronising_decompositions} Two ultimately periodic "$\omega$-words" can always be written as
		ultimately periodic "words@$\omega$-words" whose prefix (resp. period)
		have the same length.}
	\end{figure}
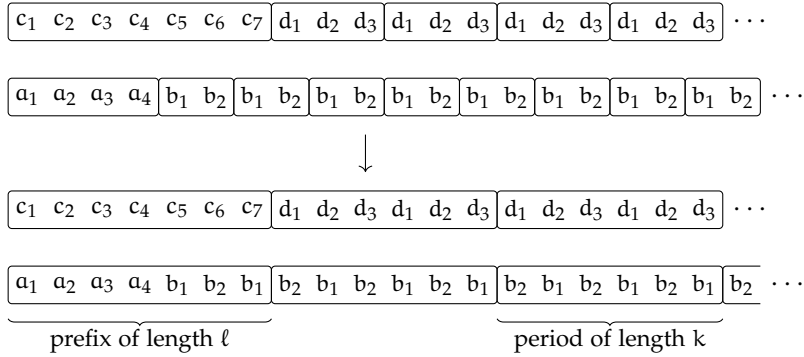

	\smallskip\noindent\emph{$\omega$-iteration case.}
	%\item 
	Assume that $X \in \Sat$ satisfies the induction hypothesis. We now show that $X\Pomega$ does so, as well. Let $(u_x)_{x \in X}$ denote the sequence computed for $X$. We will construct a sequence $(w_z)_{z \in X\Pomega}$. For each $z$ in $X\Pomega$, there exist finite words $x(z) \in X^+$ and $y(z) \in X^+$ such that $z = x(z) \cdot y(z)\Aomega$. Moreover, since $X\Pomega$ is finite, we may choose the words $x(z) \in X^+$ (resp.
	$y(z) \in X^+$) such that they all have the same length $\ell$
	(resp. $k$), as depicted in \Cref{fig:synchronising_decompositions}. After having made such a choice of $x(z)$ and $y(z)$ for each $z \in X\Pomega$, now let $z \in X\Pomega$ be arbitrary. Write $x(z) = x_1(z) \cdots x_\ell(z)$ and $y(z) = y_1(z) \cdots y_{k}(z)$, where $k$ may depend on $z$, but $\ell$ does not. Now define the ordinal word $w_z := u_{x_1(z)} \cdots u_{x_\ell(z)} \cdot (u_{y_1(z)} \cdots u_{y_{k}(z)})^\omega$. We get that
	\[ \product(\ordmap(w_z)) = x(z) \cdot y(z)\Aomega = z.\]
	Now let $z, z' \in X\Pomega$.  We need to show that $w_{z} \FOkeq w_{z'}$.
	Since, using the induction hypothesis, $u_{x_i(z)} \FOkeq u_{x_i(z')}$,
	we get that 
	\begin{equation}\label{eq:prefix-equiv}
		u_{x_1(z)} \cdots u_{x_\ell(z)} \FOkeq u_{x_1(z')} \cdots u_{x_\ell(z')},
	\end{equation} 
	using the first item of \Cref{prop:rosenstein-main-text}. 
	Similarly, as a consequence of the second item of \Cref{prop:rosenstein-main-text}, we obtain
	% \sam{perhaps isolate this as a consequence earlier}
	% note that if $V$ is a finite set of countable ordinal words such that $v \FOkeq v'$ for every $v, v' \in V$, then, for any $w, w' \in V^+$, we have $w^\omega \FOkeq (w')^\omega$. In particular, we obtain that 
	\begin{equation}\label{eq:suffix-equiv}
		(u_{y_1(z)} \cdots u_{y_{k}(z)})^\omega \FOkeq (u_{y_1(z')} \cdots u_{y_{k}(z')})^\omega.
	\end{equation}
	Combining \Cref{eq:prefix-equiv} and \Cref{eq:suffix-equiv}, we conclude that $w_z \FOkeq w_{z'}$.

	\smallskip\noindent\emph{Merging operator.}
	%\item
	Finally, assume that $X \in \Sat$ satisfies the induction hypothesis. We now show that $X\Pmerge$ does so, too. Let $(u_x)_{x \in X}$ denote the sequence computed for $X$. We will construct a sequence $(w_y)_{y \in X\Pmerge}$. By \Cref{prop:rosenstein-main-text}, pick $n$ large enough such that $u^n \FOkeq u^{n+1}$ for all $u \in \alphabet\ordp$. 
	For each $y \in X\Pmerge$, pick $m(y) \geq n$ and $x_1(y), \dots, x_{m(y)}(y) \in X$ such that $y = x_1(y) \cdots x_{m(y)}(y)$. Define $w_y := u_{x_1(y)} \cdots u_{x_{m(y)}(y)}$. Using the fact that $\product(\ordmap(u_x)) = x$ by the induction hypothesis, it is clear that $\product(\ordmap(w_y)) = y$. Now, for any $y,y' \in X\Pmerge$, note that, for any $1 \leq i \leq m(y)$ and $1 \leq j \leq m(y')$, we have $u_{x_{i}(y)} \FOkeq u_{x_j(y')}$, using the induction hypothesis. Writing $u_0 := u_{x_1(y)}$, it follows in particular from the first item of \Cref{prop:rosenstein-main-text} that
	\[w_{y} \FOkeq u_0^{m(y)} \FOkeq u_0^{m(y')} \FOkeq w_{y'}.\qedhere\]
	%Let $y,y' \in X\Pmerge$.
	%
	% By  ... , pick $p, q \geq n$ and $x_1, \dots, x_p, x_1', \dots, x_q' \in X$ such that $x_1 \cdots x_p = y$ and $x_1' \cdots x_q' = y'$. By the induction hypothesis, pick $u \in \alphabet\ord$
	%	
		% By the induction hypothesis on $X$,
		% for every $k$, pairs of elements of $X$ can be written as
		% images of $\product\circ\ordmap$ of $\FOkeq$-equivalent words.
		% By definition of $X\Pmerge$, there exist
		% $x_1,\hdots,x_p \in X$ and $x'_1,\hdots,x'_q \in X$
		% such that $x_1\cdots x_p = y$ and $y_1 \cdots y_q = y'$, 
		% where $p,q \geq n$ where $n$ is as in
		% \Cref{prop:rosenstein-main-text}.
		% Moreover, there exists $u_1,\hdots,u_p,v_1,\hdots,v_q \in 
		% \alphabet\ord$ such that $\product(\ordmap(u_i)) = x_i$
		% and $\product(\ordmap(v_j)) = y_j$ and $u_i \FOkeq v_{j}$
		% for all $i,j$. By letting $u := u_1$, we get that
		% $u_1\cdots u_p \FOkeq u^p$ and $v_1\cdots v_q \FOkeq u^q$.
		% Since $p,q \geq n$, we get by \Cref{prop:rosenstein-main-text} that
		% $u^p \FOkeq u^q$ and hence $y$ and $y'$ can be written as
		% the image via $\product\circ \ordmap$ of $\FOkeq$-equivalent "words".\qedhere
	%\end{itemize}
\end{proof}

\section{Proofs of when the algorithm says `"yes"'}
\label{section:proof-yes}
% !TeX root = ../main-separation-ordinals.tex

\subsection{Proof of \Cref{prop:fo-concat-union}}
\label{proof-prop:fo-concat-union}
\begin{proof}
	Let~$\rho\colon K\to M$, $\tau\colon L\to M$ be the "\FO-approximants".

	For the union: We define $\sigma(u)$ to be $\rho(u)$ if~$u\in K$, and~$\tau(u)$ otherwise.
	This is clearly an "\FO-approximant" over~$K\cup L$.

	For the concatenation:
	We first establish the result when~$K,L$ do not contain the empty word.
	One defines the map~$\sigma$ wich, given a word~$u$ as input, defines the least~$x$ such that~$u|_{<x}\in K$ and~$u|_{\geqslant x}\in L$, and outputs~$\rho(u|_{<x})\Acdot\tau(u|_{\geqslant x})$. It is easy to check that~$\sigma$ is an "\FO-approximant": the definability comes from the fact that $K$ and~$L$ are "\FO-definable@@lang", and since~$\product(u|_{<x})\leqslant \rho(u|_{<x})$ and $\product(u|_{\geqslant x})\leqslant \tau(u|_{\geqslant x})$, we get
	that~$\product(u)=\product(u|_{<x})\Acdot\product(u|_{\geqslant x})\leqslant \rho(u|_{<x})\Acdot \tau(u|_{\geqslant x})=\sigma(u)$.
	Adding the case of~$\varepsilon$ is then easy done by case distinction.
\end{proof}

\subsection{Proof of \Cref{lemma:fo-approximant-omega}}
\label{proof-lemma:fo-approximant-omega}
\begin{proof}
	The proof is similar to the finite word case.
	We use a double induction on~$|\Closgp A|$ and~$|A|$.
	The induction is guided by \Cref{lemma:finite-trichotomy}.
	The base case is~$A=\varnothing$, and the nowhere defined 
	"\FO-approximant" proves it.
	% NON THOMAS
	% ne supprime pas la phrase qui suit, elle est nécessaire :
	% on se sert implicitement du cas |A| = 1 dans la preuve…
	% Éventuellement tu peux la déplacer autre part et la reformuler
	% mais il faut le mentionner explicitement.
	Moreover, note that if~$A$ has only one letter,
	then there is only one word in $A\Womega$,
	so the property trivially holds.

	\emph{First case:}~$a\Acdot \Closgp{A}\subsetneq \Closgp A$ for some~$a\in A$.
	Let~$B ::= A\smallsetminus \{a\}$.
	We provide first an "\FO-approximant" for $(a^+B^+)^\omega$.
	We have from the finite case an "\FO-approximant"~$\tau$ from~$a^+B^+$ to $a\Acdot \Closgp A$,
	We use again the "condensation \FO-formula" $\varphi(x,y)$ that
	expresses that ``two positions~$x$ and~$y$ are equivalent
	if the subword on the interval~$[x,y]$ belongs to $a^*B^*$''.
	By induction hypothesis\footnote{Indeed,
	$|\Closgp{(a\Acdot \Closgp{A})^+}| <
	|\Closgp{A}|$ by \Cref{prop:saturation-take-letter-out}.} on~$A'=a\Acdot \Closgp A$, there is an "\FO-approximant"
	from~$(a\Acdot \Closgp{A})^\omega$ to~$\Closgp{A}$. By \Cref{lemma:composition-fo-condensation},
	we thus obtain an "\FO-definable map" from~$ (a^+B^+)^\omega$ to~$\Closgo{A}$. It is an "\FO-approximant" by construction.
	Now, using the finite case, the case for
	a one-letter alphabet, the induction hypothesis on smaller
	alphabets and \Cref{prop:fo-concat-union,lemma:fo-approximant-omega}, it can be easily extended
	to an "\FO-approximant" from~$A^\omega= A^*((a^+B^+)^\omega\cup B^\omega \cup a^\omega)$ to~$\Closgo{A}$. 
		
	\emph{Second case:}~$\Closgp{A}\Acdot a \subsetneq \Closgp A$.
	This case is similar to the first one.

	\emph{Third case:}~$\Closgp A$ has a maximum~$M$.
	Then the constant map that sends every "word" in~$A^\omega$ to~$M\Aomega$
	is an "\FO-approximant".
\end{proof}

\subsection{Proof of \Cref{lemma:trichotomy-ordinal-words}}
\label{proof-lemma:trichotomy-ordinal-words}

Our proof of \Cref{lemma:trichotomy-ordinal-words}, requires some
basic properties of "Green's relations" on "ordinal monoids".

\AP Recall that a ""preorder"" on a set $X$ is a binary relation $\preceq$
over $X$ that is reflexive ($x \preceq x$ for all $x\in X$) and
transitive ($x \preceq y$ and $y \preceq z$ imply $x \preceq z$).
The ""equivalence relation associated with""
$\preceq$ is the binary relation
$\sim$ defined by $x \sim y$ if and only if $x \preceq y$ and
$y \preceq x$: it is always an equivalence relation.

 \noindent
\textbf{\AP""Green's relations""}: Define in a "monoid" $\monoid$ the following preorders:
\phantomintro\Lleq\phantomintro\Jleq\phantomintro\Rleq
\begin{align*}
  x \reintro*&\Lleq y \text{ if  $x = ay$ for some~$a\in\monoid$},&
  x \reintro*&\Rleq y \text{ if  $x = yb$ for some~$b\in\monoid$},\\
  x \reintro*&\Jleq y \text{ if  $x = ayb$ for some~$a,b\in\monoid$}.
\end{align*}
We denote by $\intro*\Jeq$, $\intro*\Leq$ and $\intro*\Req$ the "corresponding equivalence relations",
\AP and define~${\intro*\Heq} = {\Leq}\cap{\Req}$.
% and~${\intro*\Deq} ={\Leq}\circ{\Req}$ \remi{I think we don't use $\Deq$}.
\AP Given a relation $\mathcal{K} \in
\{{\Jeq}, {\Leq}, {\Req}, {\Heq}\}$,
a~\reintro*"$\mathcal{K}$-class@$\Jeq$-class" 
is one of its equivalence classes, 
\phantomintro{$\Jeq$-class}\phantomintro{$\Heq$-class}\phantomintro{$\Leq$-class}\phantomintro{$\Req$-class}\phantomintro{$\Deq$-class}%
and $\monoid$ is ""$\mathcal{K}$-trivial@greentrivial""
if all "$\mathcal{K}$-classes@$\Jeq$-class" are singletons.
\AP We assume, in the following paragraphs, that the monoid
$\monoid$ is finite.
\AP A $\Jeq$-class $J$ is ""regular@@Jclass"" when one of the following
equivalent property hold---see e.g. \cite[\S V.2.2]{pin2020mpri} for a proof:
\begin{itemize}
  \item there exists two elements of $J$ whose product stays in $J$,
  \item $J$ contains an idempotent, or
  \item all "$\Leq$@$\Leq$-classes" and "$\Req$-classes" of $J$ contains  an "idempotent".
\end{itemize}
There is an equivalent notion that expresses the fact that $J$ behaves
nicely with "$\omega$-iteration" if we assume $\ordmonoid$
to be an "ordinal monoid":
\AP we say that $J$ is ""$\omega$-stable""
when one of the following equivalent properties hold:
\begin{itemize}
  \item there exists a sequence $(x_n)_{n<\omega}$ in $J$ such that
  $\product((x_n)_{n<\omega}) \in J$,
  \item there exists an element $x \in J$ such that $x\Aomega \in J$.
\end{itemize}

\begin{proposition}
  \label{prop:la-proposition-qui-fait-le-cafe}
  In a finite "ordinal monoid", all $x$ such that 
  $x \Req x\Aomega$ are "idempotent".
\end{proposition}

\begin{proof}
  If $x \Req x\Aomega$, then
  $x = x\Aomega a$ for some $a\in\ordmonoid$, and thus $xx = x x\Aomega a = x\Aomega a = x$.
\end{proof}

The statement of following proposition comes from
\cite[Lemma 7]{colcombet2015limited}.

\begin{proposition}
  \label{prop:omega-stable-H-trivial}
  In a finite "ordinal monoid" $\monoid$, "$\omega$-stable" "$\Jeq$-classes" are "$\Heq$-trivial".
\end{proposition}

\begin{proof}
  Since $J$ is "$\omega$-stable", we have $x \Jeq x\Aomega$ for some
  $x \in J$. In particular, since $x \Rleq x\Aomega$, we obtain
  by stability---see e.g. \cite[Theorem V.1.9]{pin2020mpri}---$x \Req x\Aomega$,
  and thus $x$ is "idempotent" by \Cref{prop:la-proposition-qui-fait-le-cafe}.

  We will now show that, in fact, every element of ${\Heq}(x)$ is "idempotent".
  Let $y$ such that $x \Heq y$. First, we claim that $y \Jeq y\idem$.
  Indeed, since $x \Heq y$, we have $y \Heq x = x^2 \Jeq y^2$. So $y \Jeq y^2$, and by stability, $y^2 \Heq y$, from which it follows,
  by trivial induction on $n\in\Nats_{>0}$, that $y \Heq y^n$.
  In particular, $y\idem \Heq y \Heq x$.
  Since the $\Heq$ relation is trivial
  on idempotents---see e.g.
  \cite[Corollary V.1.5]{pin2020mpri})---,
  it follows that $x = y\idem$.
  Then $y\Aomega = (y\idem)\Aomega = x\Aomega \Jeq x \Heq y$, so
  by \Cref{prop:la-proposition-qui-fait-le-cafe}, 
  $y$ is idempotent.

  Hence, every element of ${\Heq}(x)$ is "idempotent". Since $\Heq$ is trivial
  on "idempotent" elements, it means that ${\Heq}(x)$ is trivial.
  Finally, by Green's
  lemma---see \cite[Proposition V.1.10]{pin2020mpri}---$\Heq$-classes inside
  a $\Jeq$-class are equipotent, from which it follows that the whole class
  $J$ is "$\Heq$-trivial".
\end{proof}

We are now ready to prove \Cref{lemma:trichotomy-ordinal-words}.

\begin{proof}[Proof of \Cref{lemma:trichotomy-ordinal-words}]
	Assume that the first two items do not hold. %We work with Green's relations in $\Closgordp A$. 
	
	Since the first item does not hold, by~\Cref{prop:saturation-take-letter-out}, we have  $a\Acdot \Closgordp{A}= \Closgordp A$ for all~$a\in A$.
	This implies that for all $a\in A$ and~$b\in\Closgordp{A}$, $b=a\Acdot c$ for some~$c$, ie $b\Rleq a$.
	It follows that there is a unique maximal "$\Jeq$-class"~$J$ in~$\Closgordp{A}$, which is an "$\Req$-class", and which contains~$A$.
	Since furthermore~$a\Acdot a \Acdot \Closgordp{A}= \Closgordp A$, we get that~$a\Acdot a$ belongs to~$J$. Hence~$J$ is "regular@@Jclass".
	
	Since the second item does not hold and by \Cref{prop:saturation-take-letter-out},
	\[\Closgo{A}\Closgordp A=\Closgo{A}\Closgordp{\Closgo{A}}=\Closgordp A.\]
	Hence, all~$a\in\Closgordp A$ is of the form~$b\Acdot c$ where~$b\in\Closgo{A}$, and $c\in\Closgordp A$.
	Hence $J$ is "$\omega$-stable", and thus hence "$\Heq$-trivial"
	by \Cref{prop:omega-stable-H-trivial}.
	
	It follows that the product, the $\omega$-exponent
	and the $-\Amerge$ operation of elements of $J$ stay in $J$, respectively because $J$ is "regular@@Jclass",
	because it is "$\omega$-stable" and because it is "$\Heq$-trivial" and thus "group-trivial".
	Hence $\Closgordp A \subseteq J$.

	Overall~$\Closgordp A = J$ consists of a single "$\Req$-class" which is~"$\Heq$-trivial", and hence "$\Leq$-trivial".
	Hence, for all~$x,y\in\Closgordp A$, $x\Acdot y\Jeq y$ and $x\Acdot y \Lleq y$, thus $x\Acdot y\Leq x$. By "$\Leq$-triviality",
	$x\Acdot y = y$. Finally, $y^\omega=(x\Acdot y)^\omega=x\Acdot(y\Acdot x)^\omega= x\Acdot x^\omega=x^\omega$ (using $y=x\Acdot y$ and $x = y\Acdot x$ from the previous item).
\end{proof}

\subsection{Proof of \Cref{lemma:fo-approximant-aord}}
\label{proof-lemma:fo-approximant-aord}

The proof of \Cref{lemma:fo-approximant-aord} requires a few additional properties on "Green's relations" on "ordinal monoids" and "first-order logic".

\begin{lemma}
	\label{lemma:fo-def-languages}
	For every $m,p \in \Nats$,
	the following languages over $\{a\}$
	are "\FO-definable@@lang":
	\begin{multline*}
		\{a^{\kappa} \mid \kappa \geqord \omega^m\timesord p\},\qquad
		\{a^{\kappa} \mid \kappa \greaterord \omega^m\timesord p\},\\
		\{a^{\kappa} \mid \kappa \leqord \omega^m\timesord p\} \quad\text{and}\quad
		\{a^{\kappa} \mid \kappa \lessord \omega^m\timesord p\}.
	\end{multline*}
\end{lemma}
\begin{proof}
	Since "first-order definable languages" are closed under
	complementation, we only have to prove
	that $\{a^{\kappa} \mid \kappa \geqord \omega^m\timesord p\}$
	and $\{a^{\kappa} \mid \kappa \leqord \omega^m\timesord p\}$
	are "definable in first-order logic".
	We focus on the first class of languages---the second
	can be handled using similar techniques.
	The main idea is to build a formula recognising
	$\{a^{\kappa} \mid \kappa \geqord \omega^m\timesord p\}$
	by induction on $m\in\Nats$,
	using the "finite condensation".
	\begin{itemize}
	\item For $m = 0$, we need to define ordinals greater or
		equal to $p \in \Nats$, which is trivial.
	\item Then, observe that the set of "countable ordinals"
		greater or equal
		to $\omega^{m+1} \timesord p$ is condensed by
		$\condensation_{\finite}$ to the set of "countable ordinals"
		that is greater or equal to $\omega^m\timesord p$.
		If we have a formula defining the latter set,
		we can see as an "\FO-definable function"
		$G : a\ord \to \{\top,\bot\}$.
		By letting $F: a\ordp \to \{a\}$ be the constant map,
		we obtain by \Cref{lemma:composition-fo-condensation}
		that the function
		\[G \focomp{\finite} F: a\ord \to \{\top,\bot\}\]
		is "\FO-definable@@map"---and hence, the preimage of $\top$,
		that is the set of countable ordinals greater or
		equal to $\omega^p \timesord p$ is "\FO-definable@@lang".
		\qedhere
	\end{itemize}
\end{proof}

\begin{lemma}\label{lemma:omegas-L-equivalent}
	If a finite "ordinal monoid",
	$x\Jeq y\Jeq x\Aomega\Jeq y\Aomega$ implies~$x\Aomega\Leq y\Aomega$.
\end{lemma}
\begin{proof}
	Indeed, since $(x\cdot x)\Aomega = x\Aomega\Jeq x$, we get $x\cdot x\Jeq x$.
	Hence~$x^m\Jeq x$ for all~$m$, and the same holds for~$y$;
	Hence, there is a power~$m$ such that~$x^m$ and~$y^m$ are $\Jeq$-equivalent idempotents. 
	However, we know that two idempotents in the same "$\Jeq$-class" are conjugate: there exists~$a,b$
	such that $a\cdot b=x$ and $b\cdot a=y$.
	We now have~$x\Aomega=(x^m)\Aomega=(a\cdot b)\Aomega=a\cdot(b\cdot a)\Aomega = a \cdot (y^m)\Aomega = a\cdot y\Aomega$.
	Thus, $x\Aomega\Lleq y\Aomega$. Since~$x$ and~$y$ play a symmetric role, we obtain~$x\Aomega\Leq y\Aomega$.
\end{proof}

\begin{lemma}\label{lemma:omegas-stabilise}
	In a finite "ordinal monoid", for every element $x$,
	there exist~$\ell$ such that $x^{\kl[\Aomega]{\underline\omega}^\ell}=x^{\kl[\Aomega]{\underline\omega}^{\ell+1}}$.
\end{lemma}

\begin{proof}
	Consider first some~$y$ such that $y\Req (y\Aomega)\Aomega\Req y\Aomega$.
	By applying \Cref{lemma:omegas-L-equivalent} on~$y$ and $z=y\Aomega$, we get that~$y\Aomega\Leq z\Aomega=y^{\kl[\Aomega]{\underline\omega}^2}$
	Hence~$y\Aomega \Heq y^{\kl[\Aomega]{\underline\omega}^2}$.
	Hence, by \Cref{prop:omega-stable-H-trivial}, the $\Jeq$-class of~$y$ is "$\Heq$-trivial".
	We obtain~$y\Aomega = y^{\kl[\Aomega]{\underline\omega}^2}$.
	
	Consider now the series~$n\mapsto x^{\kl[\Aomega]{\underline\omega}^n}$. It is $\Rleq$-decreasing since~$y\Aomega\Rleq\omega$.	
	In less than $2|M|$ steps, by PH, there is~$\ell$ such that $x^{\kl[\Aomega]{\underline\omega}^{\ell}}\Rgeq x^{\kl[\Aomega]{\underline\omega}^{\ell+1}}\Rgeq x^{\kl[\Aomega]{\underline\omega}^{\ell+2}}$.
	By the previous remark applied on~$y=x^{\kl[\Aomega]{\underline\omega}^{\ell}}$, we obtain that the series $n\mapsto x^{\kl[\Aomega]{\underline\omega}^n}$ is constant starting at~$\ell+1$.
\end{proof}

We are now ready to prove \Cref{lemma:fo-approximant-aord}, which states 
the existence of an "\FO-approximant" over one-letter "words".

\begin{proof}[Proof of \Cref{lemma:fo-approximant-aord}]
	Let $n \in \Nats$ be such that $x^n = x\idem$ for every
	$x \in \ordmonoid$, and following \Cref{lemma:omegas-stabilise},
	let $\ell \in \Nats$ be such that $x^{\omega^\ell} = x^{\omega^{\ell+1}}$ for every $x\in\ordmonoid$.
	We use now ``Cantor's normal form'' for "countable ordinals" (see, e.g., \cite[Thm.~2.26]{jech2006set}), 
	and get that $\kappa$ can be uniquely written as
	\begin{align*}
		\kappa = \omega^\ell \timesord \kappa_\ell \plusord
		\omega^{\ell-1}\timesord k_{\ell-1} \plusord
		\hdots \plusord
		\omega^1 \timesord k_1 \plusord
		k_0
	\end{align*}
	where $k_m$ are natural numbers for $m<\ell$ and $\kappa_\ell$
	is a countable ordinal.
	Then define $\rho\colon a\ord \to \Closgord{a}$ by:
	\begin{align*}
		\rho(a^\kappa) &:= \tau(a^{\omega^\ell\timesord \kappa\ell}) \Acdot
		\tau(a^{\omega^{\ell-1}\timesord k_{\ell-1}}) \cdots
		\tau(a^{\omega^1 \timesord k_1}) \Acdot \tau(k_0)\\
		&\text{in which}\qquad
		\tau(a^{\omega^m \timesord k_m}) := 
			\begin{cases}
				a^{\omega^m\timesord k_m} & \text{ if $k_m < n$,} \\
				(a^{\omega^m})\Amerge & \text{ otherwise,}
			\end{cases}
	\end{align*}
	for every $m < \ell$
	and $\tau(a^{\omega^\ell \timesord \kappa_\ell}) = \Aunit$
	if $\kappa_\ell = 0$ and $a^{\omega^\ell}$ otherwise.
	Note that, by definition of $\ell$, we have
	\[
		a^{\omega^\ell} = (a^{\omega^\ell})\Amerge = a^{\omega^{\ell+1}} =
		(a^{\omega^{\ell+1}})\Amerge.
	\]
	The function $\rho$ satisfies $\pi(u) \leqslant \rho(u)$
	for every $u\in a\ord$, by definition, and its 
	"\FO-definability@@lang" follows from \Cref{lemma:fo-def-languages}.
\end{proof}

\section{Proofs for the related problems}
\label{section:proof-related}
% !TeX root = ../main-sepration-ordinals-fossacs.tex
% !TEX root = ../main-sepration-ordinals-fossacs.tex

In this appendix, we first prove
the computability of "\FO-pointlikes" (\Cref{proposition:pointlikes-computable})
and then deduce the decidability of the "\FO-covering problem"
for "countable ordinal words".

\subsection{Proof of \Cref{proposition:pointlikes-computable}}
\label{proof-proposition:pointlikes-computable}

We have to prove that a set~$X$ is "pointlike" if and only if it is contained in some~$Y\in\Sat$.

Let us first show that all sets~$X\in\Sat$ are "pointlike".
Let us fix some~$k$, and set~$(u_x)_{x\in X}$ to be the sequence of words which exists from 
\Cref{lemma:saturation-sets-are-pointlike} for the set~$X$ and the value~$k$.
Set also~$u=u_x$ for some arbitrarily chosen $x\in X$.
Since~$u_x\FOkeq u$ for all~$x\in X$, $u_x\in\closureFOk{u}$. Hence 
\begin{align*}
	X&=\{\product(\ordmap(u_x))\colon x\in X\}
		\subseteq \product(\ordmap(\closureFOk{u}))\ .
\end{align*}
Since this holds for all~$k$, we obtain~$X\in\PL(\lettermap)$.

Conversely, let~$X$ be a "pointlike", we aim at proving that $X\subseteq Y$ for some~$Y\in\Sat$.
For this, let~$\PFOproduct$ be the "\FO-approximant" which exists from \Cref{lemma:completude-core}.
Since~$\PFOproduct$ is an "\FO-definable map", there exists some~$k$ such that all~$\PFOproduct^{-1}(x)$ with $x\in X$ is
defined by an "\FO-sentence" of "quantifier depth" at most~$k$. By definition of "pointlike sets", there exists a word~$u\in \alphabet\ord$
such that $X\subseteq \product(\ordmap(\closureFOk{u}))$. In particuliar, this means that for all~$x\in X$, there exists~$u_x\FOkeq u$
such that~$\product(\ordmap(u_x))=x$. Let~$Y=\PFOproduct(\singordmap(u))\in\Sat$. Since~$u_x\FOkeq u$ and~$k$
has been chosen sufficiently large with respect to~$\PFOproduct$, we have~$\PFOproduct(\singordmap(u_x))=Y$ for all~$x\in X$.
We obtain:
\begin{align*}
	X	&\subseteq\{\product(\ordmap(u_x))\colon x\in X\}\\
		&\subseteq\down\{\PFOproduct(\singordmap(u_x))\colon x\in X\}\\
		&=\down Y\in\down\Sat\ .
\end{align*}

\subsection{Proof of \Cref{proposition:covering-decidable}}
\label{proof-proposition:covering-decidable}

	One begins, as usual with a finite "ordinal monoid" \(\ordmonoid\) 
	and a map \(\lettermap: \alphabet \to M\) that recognises the languages~$L,K_1,\dots,K_n$.
	\AP The ""algorithm@@cov"" computes $\Sat$ and answers~`""no@@cov""' if there exists  some $X\in\Sat$ that intersects all of
	\[
		F_L:=\product(\ordmap(L)),\ F_{K_1}:=\product(\ordmap(K_1)),\dots,
		\ F_{K_n}:=\product(\ordmap(K_n))\ .
	\]
	Otherwise, it answers~`""yes@@cov""'.
	
	Let us prove that this algorithm is correct.
	Assume that "it@algorithm@cov" answers `"yes@@cov"'.
	 In this case, let~$\rho\colon A\ord\to \Sat$ be the "\FO-approximant" that exists by \Cref{lemma:completude-core}, where $A = \{\{\lettermap(a)\} \mid a \in \alphabet\}$, and define
		\begin{align*}
			C_i &= \{u\in\alphabet\ord\mid\rho(\singordmap(u))\cap \product(\ordmap(K_i))=\varnothing\}\ .
		\end{align*}
		Let~$u\in L$. We know that~$\rho(\singordmap(u))\in \Sat$
		intersects~$\product(\ordmap(L))$, and thus, there is some~$\product(\ordmap(K_i))$ that it does not intersect. Hence~$u\in C_i$.
		Furthermore, by construction, if~$u\in C_i$, then~$\rho(\singordmap(u))\cap\product(\ordmap(K_i))=\varnothing$ by definition, and thus, since $\product(\ordmap(u))\in\rho(u)$, $\product(\ordmap(u))\not\in\product(\ordmap(K_i))$. Hence~$C_i\cap K_i=\varnothing$. Hence, there is a positive answer to the "covering problem".	
	
	If the "algorithm@@cov" answers `"no@@cov"', this is because there exists~$X\in\Sat$
	that intersects all of $F_L,F_{K_1},\dots,F_{K_n}$.
	Let $x\in X\cap F_L$,
	and~$x_i\in X\cap F_{K_i}$ for all~$i=1\dots n$ be the elements witnessing these intersections.
	Assume for the sake of contradiction that there would exist~$C_1,\dots,C_n$
	witnessing that the answer to the "covering problem" is positive. Let~$k$ be sufficiently large for $C_1,\dots,C_k$
	be all definable by "\FO-sentences" of "quantifier depth" at most~$k$.
	 By \Cref{lemma:saturation-sets-are-pointlike}, there exist $\FOkeq$-equivalent words~$u_x$ for all~$x\in X$ such that~$\product(\ordmap(u_x))=x$ for all~$x\in X$.
	 Since~$L$ is "recognised@@OM" by $(\ordmonoid,\lettermap,F_L)$ and $K_i$ by $(\ordmonoid,\lettermap,F_{K_i})$,
	 it follows that $u_x\in L, u_{x_1}\in K_1,\dots, u_{x_n}\in K_n$.
	 Since~$u_x\in L$, and $L\subseteq\cup_i C_i$, it follows that there exist some~$i$ such that $u_x\in C_i$.
	 But~$C_i$ is defined by a "\FO-sentence" of "quantifier depth" at most~$k$ and $u_x\FOkeq u_{x_i}$, thus $u_{x_i}\in C_i$.
	 This witnesses that~$C_i\cap K_i\neq\varnothing$. A contradiction.

\end{document}